\documentclass[10pt,journal,compsoc]{IEEEtran}
\usepackage[T1]{fontenc}

\ifCLASSINFOpdf

\else
\fi

\usepackage{graphicx}
\usepackage{algorithm}
\usepackage{epstopdf}
\usepackage{footnote}
\usepackage{tikz}
\usepackage{subfigure}
\usepackage{cuted}

\graphicspath{{./Figures/}}

\usepackage[cmex10]{amsmath}
\usepackage{amsfonts,latexsym,amssymb,amsthm}
\usepackage{epsfig,graphics,color}
\usepackage{setspace}
\usepackage{multirow}
\usepackage{cite}
\usepackage{algorithm, algpseudocode}
\usepackage{mathrsfs}
\usepackage{bm}
\usepackage{enumitem}
\usepackage{flushend}
\setlist[itemize]{noitemsep, topsep=0pt}

\newtheorem{theorem}{Theorem}
\newtheorem{lemma}{Lemma}

\newtheorem{proposition}{Proposition}

\newtheorem{remark}{Remark}

\usepackage{soul}

\begin{document}
	%
	\title{In-network Computation for Large-scale Federated Learning over Wireless Edge Networks}
	%
	%
	%
	
	\author{Thinh~Quang~Dinh,  Diep~N.~Nguyen,  Dinh Thai Hoang,  Pham Tran Vu, and Eryk Dutkiewicz 
			\thanks{Preliminary results in this paper are presented at the IEEE GLOBECOM Conference, 2021~\cite{Globecom01}.}
	}
	
	%
	%




	\maketitle
	
	\begin{abstract}
		 Most conventional Federated Learning (FL) models  are using a star network topology where all users aggregate their local models at a single  server (e.g., a cloud server). That causes significant overhead in terms of both communications and computing at the {server}, delaying the training process, especially for large scale FL systems with straggling nodes. This paper proposes a novel edge network architecture that {enables} decentralizing the model aggregation process at the server, thereby significantly reducing the training delay for the whole FL network. Specifically, we design a highly-effective in-network computation framework (INC) consisting of {a user scheduling mechanism, an in-network aggregation process (INA) which is designed for both primal- and primal-dual methods in distributed machine learning problems, and a network routing algorithm with theoretical performance bounds}. The in-network aggregation process, which is implemented at edge nodes and cloud node, can adapt two typical methods to allow edge networks to effectively solve  the distributed machine learning problems. Under the proposed INA, we then formulate a joint routing and resource optimization problem, aiming to minimize the aggregation latency. The problem turns out to be  {NP-hard, and thus we propose} a  {polynomial time} routing algorithm   which can achieve near optimal performance with a theoretical bound. Simulation results {showed} that the proposed algorithm  {can achieve} more than 99$\%$ of the optimal solution and {reduce} the FL training latency, up to $5.6$ times w.r.t other baselines. The proposed INC framework can not only help reduce the FL training latency but also significantly decrease cloud's traffic and computing overhead. By embedding the computing/aggregation tasks at the edge nodes and leveraging the multi-layer edge-network architecture, the INC framework can liberate FL from the star topology to enable large-scale FL.
		\end{abstract}

	%
	\IEEEpeerreviewmaketitle

	\section{Introduction} \label{sec:intro}
	
	Machine learning that enables intelligent systems {has} emerged as a key technology  benefiting many aspects of modern society \cite{SaadNetMag2020}. Currently, most data-driven systems collect  {data from local devices and then train the data} at centralized cloud servers. 
	However, user privacy, applications' latency  {and network overhead traffic} are major concerns of this centralized solution.
	To address these issues, collaborative learning schemes  {have} received considerable attention where mobile users (MUs) can build and share their machine learning models instead of sending their {raw data} to a centralized data center \cite{LimCOMST2020}. Out of these distributed learning schemes, Federated Learning (FL) has recently emerged as a promising candidate \cite{McMahan2017}. In  FL, each MU iteratively computes  {its} local model using  {its} local training data.  {This local model} is then  {sent to a cloud server and aggregated together with other local models}, contributed by other MUs to produce a global model.  {The global model is then sent back to all the MUs and then these MUs will use it to update their new local models accordingly.}  This  process is  repeated   {until it is converged or}  an  accuracy  level  of  the  learning  model  is reached. Since no user's raw data is exchanged, FL facilitates machine learning in many aspects from data storage, training to data acquisition, and privacy preservation.
	
	Although communications overhead is significantly saved by sharing local models instead of local raw data, communication cost is still a major bottleneck of FL. This is because to aggregate the global model, the centralized server generally needs to connect with a huge number of users for all their updates via star network topology \cite{Wiedemann2020,lin2018deep,TianyiNIPS2018,McMahan2017,XiaTWC2020,YangTC2020,ChenICC2020}. Moreover, for advanced deep learning models \cite{Bert2019,Simonyan15} which contain a vast amount of parameters, the size of models exchanged in networks could be very large. For example, the BERT model \cite{Bert2019} in Natural Language Processing area is up to $1.3$ GB, while VGG16 \cite{Simonyan15} in Computer Vision is more than $500$ MB.  {Thus, with} VGG16, it can cost about $500$ TB for each user until the global model is converged \cite{SamekTut2020}. {Consequently}, high communication cost can lead to (a) high transmission latency and (b) high traffic at the cloud server. To address the issues, new network architectures are needed be investigated.

	High computation cost is another challenge  of conventional FL models using the star network topology. Let us assume that vectors or matrices are used to store users' models for each FL iteration. With VGG16, vectors or matrices contain more than $138$ million elements. As a result, when the number of users grows, aggregation operations at the centralized server can be very computationally costly \cite{Qian2012,Eric2017} in terms of processing and memory resources. 
	For that, novel network topologies should be considered to enable large-scale FL systems.

	Edge Computing (EC), an emerging distributed network architecture \cite{LimCOMST2020} that aims to  {bring cloud functions and resources (including computing, storage and networking capacities) closer to end users} can be a promising solution to large scale FL. Since edge nodes (EN) possess both computation and communication capacities, edge  {networks} can  {effectively support} the cloud server to decentralize its communication and computing burden of model aggregation in very large scale FL networks. 
	To exploit potential benefits of edge networks for  {FL models}, it is necessary to develop a distributed in-network aggregation (INA) functionality that can be implemented at edge nodes of EC to liberate FL from the conventional start network topology.

	In-network computation (INC) is   a   process   of gathering and routing information through a multi-hop network, then processing data at intermediate nodes with the objective of optimizing the resource consumption \cite{FasoloWC2007}. Currently, the concept is well-studied for big data analytics such as MapReduce \cite{DeanMapReduce}, Pregel \cite{GrzegorzPregel} and
	DryadLINQ \cite{YuanDryadLINQ} for distributed data clusters. Three  basic components of an in-network computation solution are: suitable networking protocols, effective aggregation functions and  efficient  ways of representing data \cite{FasoloWC2007}. However, INC frameworks for FL is under-studied. Moreover, edge nodes will be densely deployed in future networks \cite{ChenJSAC2018}, one can exploit this diversity by considering a scenario where
	 {an} MU can associate with multiple nearby edge nodes, instead of only a single edge node. Under such a scenario, the problems of network routing and network resource allocation become more challenging.

Regarding machine learning models, e.g., classification or regression, they can be  {learned} by solving regularized loss minimization problems in two ways: (a) by using primal methods such as Stochastic Gradient Descent (SGD) \cite{McMahan2017} or (2) by using the primal-dual methods such as Stochastic Dual Coordinate Ascent (SDCA) \cite{Ma2017,NIPS2014_894b77f8}.  Most works in FL are limited in one method \cite{Wiedemann2020,lin2018deep,TianyiNIPS2018,McMahan2017,XiaTWC2020,YangTC2020,ChenICC2020,YangTWC2020,CaoTWC2020,TranINFOCOM2019,WangJSAC2019,LiuICC2020}. {From experimental and theoretical results  \cite{YangTC2020,Shwartz2013}, it is noted that the primal methods cost less resources, e.g. computing and storage, and can faster reach a moderate accuracy level, while the primal-dual method can provide better convergence and accuracy in the long run.
As a result,} to achieve high accuracy with faster convergence, instead of using vanilla SDCA, authors in \cite{Shwartz2013} implemented SGD in initial rounds before executing SDCA. Therefore, a general in-network computation protocol, which can adapt to  {both the FL} methods, is required.
	
	
Given the above, this paper proposes a novel edge network architecture and its associated in-network computation framework that {enable} decentralizing the model aggregation process at the server, thereby significantly reducing the training delay for the whole FL network. To this end, a highly-effective in-network computation protocol is proposed with three main components:  {a user scheduling mechanism, a novel in-network aggregation process which is designed for both primal- and primal-dual methods in distributed machine learning problems, and a network routing algorithm with theoretical performance bounds}. 
	The major contributions of this work are summarized as follows.
	\begin{enumerate}
		\item Propose a novel edge network architecture to  decentralize the communication and computing burden of cloud node in FL aiming at minimizing the whole network latency. Such a network architecture is enabled by a novel in-network computation protocol consisting of an effective in-network aggregation process, a routing algorithm and a user scheduling.

		\item Design a new user scheduling mechanism which allows a group of users to be able to send their models in advance instead of waiting for the last user finishing its local update. Based on the power law distribution of user data, this mechanism could mitigate straggler effect with a theoretical performance bound.
		
		\item  Develop a novel in-network aggregation process that instructs on how data are processed at edge nodes and cloud node to decentralize the global model aggregation process. The INA is designed for both primal- and dual-primal solutions to a given FL network.
		
		\item  Optimize the joint routing and resource allocation for MUs and ENs under the proposed architecture, thereby minimizing FL training time of the whole network. We show that the resulting mixed integer non-linear programming problem is NP-hard then propose an effective solution based on randomized rounding techniques. Simulations show that the proposed solution can achieve more than 99$\%$ of the optimal solution.		
		
	\end{enumerate}
	
		The rest of the paper is organized as follows. In Section \ref{section:related_work}, we present the related works. The system
		model is introduced in Section \ref{section:sys_model}. Then, the user scheduling mechanism is proposed in Section \ref{sec:user_shed}. Section \ref{sec:ina} presents our in-network aggregation process.  Next, we formulate the network routing and resource allocation frameworks in Section \ref{section:problem_for}, and propose our solution in Section \ref{section:solution}.   We then present the numerical results in Section \ref{section:numerical} and final conclusions are drawn in Section \ref{section:conclusion}.
	
	\section{Related Works} \label{section:related_work}
	
	Although sharing raw local data is not required, high communication cost still remains a major obstacle in FL systems, especially given its canonical star network topology. Solutions to this issue can be divided into several directions such as compressing the local models using quantization, e.g., 
		\cite{Wiedemann2020,lin2018deep}, skipping unnecessary gradient calculations or global updates \cite{TianyiNIPS2018,McMahan2017}, and selecting a promising user subset of each global update \cite{XiaTWC2020,YangTC2020,ChenICC2020}. For  {compressing local models},  MPEG-7 part 17 \cite{Wiedemann2020} has become a universal neural network compression standard, which is not only limited in the domain of multimedia. However, among these works \cite{Wiedemann2020,lin2018deep,TianyiNIPS2018,McMahan2017,XiaTWC2020,YangTC2020,ChenICC2020}, authors usually assumed that MUs can be directly connected to a single server. This star network topology assumption is impractical since it does not reflect the hierarchical structure of large-scale wireless networks. Moreover, for a large number of users, such a star topology cannot help scale up the FL system. 
		
		Edge computing is a potential infrastructure which can help to address the high communication cost of FL. Even though there are many works where  {edge networks are complementary to the cloud} \cite{JiaoTON2017Smoothed,PoularakisTON2020,DinhTWC2020, 8647856}, the capacities of edge networks are under explored in  existing FL works \cite{Wiedemann2020,lin2018deep,TianyiNIPS2018,McMahan2017,XiaTWC2020,YangTC2020,ChenICC2020,YangTWC2020,CaoTWC2020,TranINFOCOM2019,WangJSAC2019}. To decentralize and redistribute the model aggregation process of FL at cloud server to edge nodes in edge networks, this article leverages the ``in-network computing" concept. Although such a concept has been well-studied for traditional machine learning paradigms where users' raw data are collected and stored at big data clusters \cite{DeanMapReduce,GrzegorzPregel,YuanDryadLINQ}, INC for distributed learning paradigms, e.g., FL, especially with the aid of edge computing, has not  been visited. 
		
		There have been early works proposing decentralizing aggregation processes for different network architectures such hierarchical topology \cite{LiuICC2020,HosseinalipourTON2022}, ring topology \cite{LeeAAAI2021}, and random graph \cite{Bellet2021}. However, in those works, the proposed decentralized aggregation processes can be considered as variants of either FedAvg or CoCoA.  In our paper, the proposed in-network aggregation process, a component of our INC solution, can adapt to both primal and primal-dual methods. Moreover, existing works mostly focused on designing aggregation processes while network resource management and routing problems were not jointly considered. For example, in \cite{LiuICC2020}, users are assumed to connect with a single edge node without any alternative links. In reality with dense edge networks, MUs can associate to more than one edge nodes. Such a practical scenario calls for optimal network routing and resource allocation solution, aiming at minimizing the system latency. 
		In this article, under the proposed INC, we formulate a joint routing and resource optimization problem, aiming to minimize the aggregation latency. The problem turns out to be  {NP-hard, and thus we design} a  {polynomial time} routing algorithm   which can achieve near optimal performance with a proven theoretical bound. Last but not least, to address slow workers that stagnate the learning systems in these existing works \cite{Wiedemann2020,lin2018deep,TianyiNIPS2018,McMahan2017,XiaTWC2020,YangTC2020,ChenICC2020,YangTWC2020,CaoTWC2020,TranINFOCOM2019,WangJSAC2019,LiuICC2020, Yuris1, Yuris2} (known as the straggler effect), we propose a user scheduling scheme that is the third component of our INC solution. In short, the key novelty of our work is to leverage edge networks to enable large-scale FL and liberate FL from its canonical start topology. All aforementioned works that aim to decentralize the FL process did not leverage edge networks nor address the associated challenges (e.g., routing, resource allocation, and the straggling effect).		
	
	\section{System Model} \label{section:sys_model}
	\begin{table*}
		\caption{Notation Used Throughout the Paper} \label{notation}
		\begin{center}
			\renewcommand{\arraystretch}{1.1}
				\begin{tabular}{p{0.1\linewidth}p{0.35\linewidth}p{0.1\linewidth}p{0.35\linewidth}}
					
					\hline
					{\bf Notation} & {\bf Definition} &{\bf Notation} & {\bf Definition}  \\ \hline
					$k$ & index of MU & $D$ & the data size of the global model \\	
					$\mathcal{K}$ & set of users & $W^d$ & the downlink data rate capacity \\
					$\mathcal{D}_k$ & local dataset of user $k$ & $T^d$ & the downlink latency for broadcasting global model\\
					$n_k$ & number of data points of user $k$ & $a_{km}$ & the aggregation routing variable \\
					$n$ & total number of data points of all users & $r_{km}$ & the uplink data rate allocation variable\\
					$i$ & index of data point & $B^{\mathrm{fr}}_m, B^{\mathrm{bk}}_m$ &  the uplink fronthaul and backhaul data rate capacity of edge node $m$	\\	
					$m$ & index of EN & $T^{\mathrm{u,\rm{fr}}}_m$ & the uplink fronthaul latency of edge node $m$ and its users \\
					$\mathcal{M}$ & set of ENs & $\gamma_m$ & the transmission latency between edge node $m$ and the cloud node \\
					$t$ & index of time & $T^{\mathrm{u}}_m$ & the uplink latency of users associated with edge node $m$ \\
					$\mathbf{w}^t$ & the global model parameter of FL in primal at learning round $t$ & $T^{\mathrm{u}}$ & the total uplink latency of the whole network\\
					$\mathbf{w}_k^t$ &  the local model parameter of FL at user $k$ in primal at learning round $t$ & $t_{k}^{\rm{cp}}$ & local training update at MU $k$ \\
					$\mathbf{v}^t$ & the global model parameter of FL in primal-dual formulation at learning round $t$ & $T$ & the training time in one learning round\\
					$\Delta \mathbf{v}_k^t$  &  the local model parameter of FL at user $k$ in primal-dual formulation at learning round $t$ & $s$ & user scheduling scheme\\
					$\psi$ & the generalized global model of $\mathbf{w}$ and $\mathbf{v}$ & $\mathcal{P}_1,\mathcal{P}_2$ & user partition $1$ and $2$\\
					$\bm{\alpha}$ & the dual variables of FL in in primal-dual formulation & $j$ & index of user partition \\
					$P(\mathbf{w})$ & the global loss function & $z$ & the indicator if FL system is using primal-dual method\\
					$l_i$ & the loss function at data point $i$ & $\bm{\phi}_k^{t}$ & the local message that user $k$ sends in learning round $t$\\
					 $r$ & a deterministic penalty function & $\bm{\varphi}_{m,j}^t$ & the message that EN $m$ sends to cloud for user partition $\mathcal{P}_j$ at learning round $t$\\
					$\xi$ & the regularizing parameter & $\bm{\lambda}^t_j $ & the weight and the parameters of the aggregated model for user partition $\mathcal{P}_j$ at learning round $t$\\
					$G(\bm{\alpha} )$ & the Fenchel-Rockafeller dual form of the global loss function $P(\mathbf{w})$ & & \\
					
					\hline
				\end{tabular}

		\end{center} 
	\end{table*}%
	Let us consider a set of $K$  mobile users {MUs}, denoted by
	$\mathcal{K} =
	\{1,\ldots,K\}$ with local datasets $\mathcal{D}_k  = \{ \mathbf{x}_i \in \mathbb{R}^d, y_i \}_{i=1}^{n_k}$ with $n_k$ data points. These MUs participate in a distributed learning process by learning a shared global model from and sharing their local models (without sharing their raw data) with a cloud server. These MUs are co-located and supported by an MEC network consisting of a set of $M$ edge nodes (ENs), denoted by $\mathcal{M} = \{1, \ldots ,M\}$. The edge nodes, controlled by the MEC operator, can be small cell base stations with their own communications and computing capacities \cite{PoularakisTON2020}. ENs are connected with a cloud server via a macro base station. Without loss of generality, we can consider the cloud server as a special EN, i.e., EN $0$, co-located with the macro base station. MUs can connect to the cloud server either through ENs or directly with the macro base station. Each MU can be associated with one or more ENs \cite{Yuris3}.  {For} example, in Fig. \ref{fig::mqms}, MU $4$, which lies in the overlapping coverage areas of EN $1$ and EN $2$, can upload its local model to either one of these ENs.
	
	\begin{figure}[t]
		\centering
		{\includegraphics[width=\linewidth]{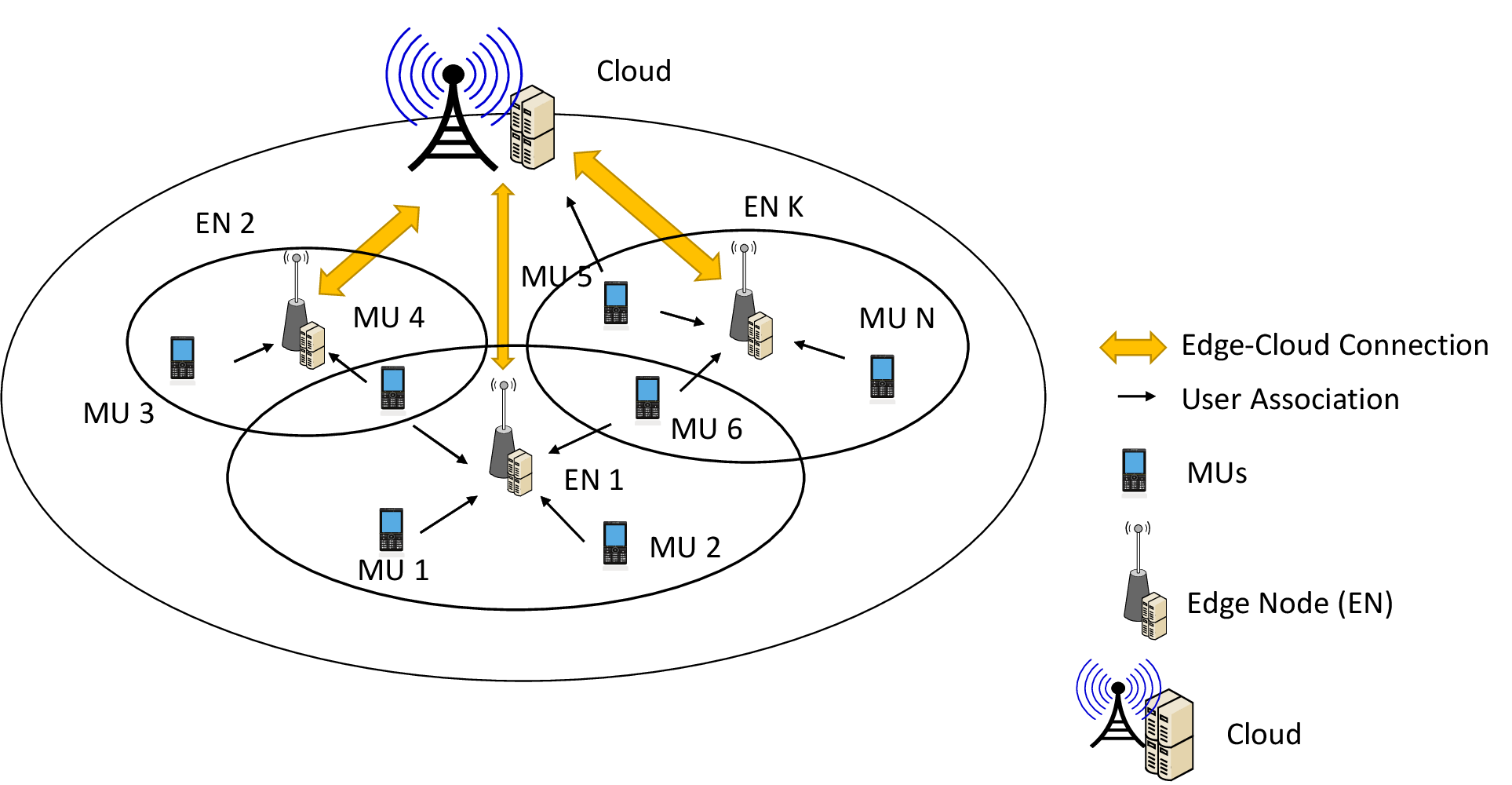}}
		\caption{FL-enabled Edge Computing Network Architecture. \label{fig::mqms} }
	\end{figure}
		
	\subsection{Federated Learning}
	To construct and share the global model, the goal is to find the model parameter $\mathbf{w} \in \mathbb{R}^d$ which minimizes the following global loss function in a distributed manner where data is locally stored at MUs:
		\begin{align}
		\min_{\mathbf{w} \in \mathbb{R}^d} \Bigg\{ P(\mathbf{w}) = \frac{1}{n} \sum_{i=1}^{n} l_i (\mathbf{x}_i^T\mathbf{w})  + \xi r(\mathbf{w}) \Bigg\},  \label{eqn:ML_primal_prob}
		\end{align}
		where $n = \sum_{i=1}^{K} n_k $, $\xi$ is the regularizing parameter and $r(\mathbf{w})$ is a deterministic penalty function \cite{zhang2010}. Here, $l_i$ is the loss function at data sample $i$.  To solve the problem (\ref{eqn:ML_primal_prob}), one can either implement the (a) primal only method or the (b) primal-dual method. For example, FedAvg is a variant  of the former one \cite{McMahan2017}, while CoCoA is a variant of the latter \cite{NIPS2014_894b77f8}. Since both FedAvg and CoCoA are widely adopted, e.g., \cite{McMahan2017,NIPS2014_894b77f8,XiaTWC2020,YangTC2020,ChenICC2020,YangTWC2020,CaoTWC2020,TranINFOCOM2019,WangJSAC2019}, in  this paper, we consider both of them in our design.
	
		For FedAvg \cite{McMahan2017}, at each iteration $t$, the cloud broadcasts $\mathbf{w}^t$ to all MUs. Based on the latest global model, each MU 
		{learns}  {its} local parameters $\mathbf{w}_k^t$ according  {to the} Stochastic Gradient Descent update
		rule aiming  at minimizing the objective function $P(\mathbf{w})$ by only using local information and the parameter values
		in $\mathbf{w}_k^t$ \cite{McMahan2017}:
		\begin{align}
		\mathbf{w}_k^t = \mathbf{w}_k^{t} - \eta (\nabla l_i(\mathbf{w}_k^{t} 
		) + \nabla r(\mathbf{w}^{t})). \label{eqn:SGD}
		\end{align}
		Here, $\nabla$ stands for partial derivative computation.
	The resulting local model updates are forwarded to the cloud for computing the new global model as follows:
		\begin{align}
		\mathbf{w}^{t+1}  = 
		\frac{1}{n} \sum_{k=1}^{K} n_k \mathbf{w}_k^t.
		\end{align}	
		
		For CoCoA, the goal is to find the model parameter $\mathbf{v} \in \mathbb{R}^d$ which is related to model $\mathbf{w}$ via the dual relationship. Following \cite{Ma2017}, using the Fenchel-Rockafeller duality, we rewrite the local dual optimization problem of (\ref{eqn:ML_primal_prob}) as follows:
		\begin{align}
		\max_{\bm{\alpha} \in \mathbb{R}^{n}}\Bigg \{ G(\bm{\alpha} ) = - \sum_{i=1}^{n}  \frac{l^*_i(-\alpha_i)}{n}  - 
		\xi r^*  (\frac{1}{\xi n} \mathbf{X}\bm{\alpha})\Bigg \}, \label{eqn:ML_dual_prob}	 
		\end{align}
		where $\{\alpha_i \}_i^n \in \mathbb{R}$ represents the set of the dual variables, $\mathbf{X} = [\mathbf{x}_1, \mathbf{x}_2, \ldots, \mathbf{x}_n] \in \mathbb{R}^{d \times	 n} $ is the total data set, $l^*_i(\cdot)$ and
		$r^*(\cdot)$ are the convex conjugate functions of $l_i(\cdot)$ and $r(\cdot)$.
		Following \cite{Ma2017}, $l_i(\cdot)$  {is} assumed to be convex with $1/\mu$-smoothness, $r(\cdot)$ is assumed to be 1-strongly convex. 
		At each iteration $t$, the cloud broadcasts $\mathbf{v}^t$ to all MUs. Based on Stochastic Dual Coordinated Ascent method, each MU  {updates} a mapping of dual variables $\Delta \mathbf{v}_k^t$ as follows:
		\begin{align}
		\Delta \mathbf{v}_k^t = \frac{1}{\xi n}\mathbf{X}_{[k]} \mathbf{h}_{[k]}^t, \label{eqn:SDDA}
		\end{align}	
		where $\mathbf{X}_{[k]}$ is denoted for
		the matrix consisting  {of only} the columns corresponding to data samples $i \in \mathcal{D}_k$, padded with zeros in all other columns, $\mathbf{h}_{[k]}^t$ is denoted for the iterative solution obtained by solving the approximated problem of (\ref{eqn:ML_dual_prob}) using only data samples $i \in \mathcal{D}_k$, which is defined in \cite{Ma2017}. 	The new global model is then aggregated at the cloud as follows:
		\begin{align}
		 \mathbf{v}^{t+1} = \mathbf{v}^{t} + \frac{1}{K} \sum_{k=1}^{K} \Delta \mathbf{v}_k^t.
		\end{align}
	
	We summarize the general procedure of these two FL methods as follows:
		\begin{figure*}[t]
			\begin{equation}
			\bm{\psi}^{t+1}  = \begin{cases}
			\frac{1}{n} \sum_{k=1}^{K} n_k \mathbf{w}_k^t & \text{if we choose the primal method}, \\
			\bm{\psi}^{t}  + \frac{1}{K} \sum_{k=1}^{K} \Delta \mathbf{v}_k^t & \text{if we choose the primal-dual method}. 
			\end{cases} \label{eqn:global_model}
			\end{equation}
		\end{figure*}	
	\begin{itemize}
		\item[1] \emph{Global Model Broadcasting}: The cloud broadcasts the latest global model to the MUs, either $\mathbf{w}^t$ or $\mathbf{v}^t$.
		\item[2] \emph{Local Model Updating}: Each MU performs local training following either (\ref{eqn:SGD}) or (\ref{eqn:SDDA}). 
		\item[3] \emph{Global Model Aggregation}: Local models are then sent back to the cloud.  Let $\psi^{t}$ denote the generalized global model. The new value of global model is then computed at the cloud by following (\ref{eqn:global_model}).
		\item[4.] Steps 1-3 are repeated until convergence.
	\end{itemize}

	\subsection{Communication Model}
	We then introduce the communication model for multi-user
	access.  For each FL iteration, the cloud node will select a set of users $\mathcal{K}^t$ at each iteration $t$. Here, how to  {select the best MUs at each learning round} is out of the scope of this paper \footnote{The learner selection in FL can be based on the quality or significance of information or location learners \cite{XiaTWC2020,YangTC2020,ChenICC2020}.}.  All users consent about their models' structure, such as a specific neural network design. Hence, let $D$ denote the data size of model parameters: 
	\begin{align}
	D = (d +1) \times \mathrm{Codeword~length},
	\end{align} 
	where $d$ is defined as the number of parameters mentioned  {as the length of $\mathbf{w}_k^t$ or $\Delta \mathbf{v}_k^t$}. 
	\subsubsection{Global Model Broadcasting}	
	Since the downlink data rate capacity $W^d$ of the  {cloud node} is much larger than that of an edge node, all users  {will} listen to the  {cloud node} at the model broadcasting step.  The latency for broadcasting the global model is $T^{\mathrm{d}} = \frac{D}{W^d}$.
	\subsubsection{Global Model Aggregation}	
	
	Let $a_{km} $ be the aggregation routing variable, where
	\begin{align}
	a_{km}  = \begin{cases}
	1 & \text{if MU $k$'s local model is directly sent to edge}\\
	& \text{node $m$}, \forall m \in \mathcal{M},\\
	0 & \text{otherwise}. 
	\end{cases}
	\end{align}
	Let $\mathbf{a}_{m}   =  [a_{1m} , a_{2m} , \cdots, a_{Km}  ]^T$  {denote} the uplink association vector of edge node $m$. We denote $\mathbf{A}  = \{a_{km} \}  \in \{0,1 \}^{K\times(M+1)}$  to be the uplink routing matrix, and $\tilde{\mathbf{a}}  = [\mathbf{a}_{0}^T,\mathbf{a}_{1}^T, \ldots , \mathbf{a}_{M}^T]^T$ to be the column vector corresponding to $\mathbf{A} $. Let $\mathcal{K}_m^t$ denote the set of user associated with edge node $m$  {at iteration $t$},  {then we have} $\bigcup_m ~ \mathcal{K}_m^t = \mathcal{K}^t$, and $|\mathcal{K}^t| = \sum_{m}|\mathcal{K}_m^t|$, where $| \cdot |$  denotes the cardinality of a set.  {Here, $\mathcal{K}_m^t$ may change over $t$ to adapt with the change of wireless channels.}

	Let $r_{km}$ denote the uplink data rate between user $k$ and edge node $m$. For ease of exposition, let $r_{k0}$ denote the uplink data rate between the cloud node and user $k$. 	Let $\mathbf{r}_{m}   =  [r_{1m} , r_{2m} , \ldots, r_{Km}  ]^T$ denote the uplink bandwidth allocation vector corresponding to edge node $m$, with $m=0$ for cloud node. 	We denote $\mathbf{R}  = r_{km} \in \mathbb{R}^{K\times(M+1)}$ as the uplink bandwidth allocation
	matrix.
	
	Let $B^{\mathrm{fr}}_m$, and $B^{\mathrm{bk}}_m$ denote the uplink fronthaul and backhaul data rate capacity of edge node $m$.  Then, uplink communication latency between edge node $m$ and  {its associated users} is the longest latency of a given user:
	\begin{align}
	T^{\mathrm{u,\rm{fr}}}_m & = \max_{k \in \mathcal{K}_m} \Bigg\{D\frac{a_{km}}{r_{km}} \Bigg \}, ~\textrm{where}~ r_{km} \leq B^{\mathrm{fr}}_m, \forall m \in \mathcal{M} \setminus \{0\}.
	\end{align}
	After edge nodes receive local models, there are two methods considered in this paper. Each edge node can help the cloud node to aggregate the local model, then send the aggregated result to the cloud node.  {Alternatively,} edge  nodes just forward received models to the cloud. Let $\gamma_m$ denote the transmission latency between edge node $m$ and the cloud node. Without in-network computation protocols, $\gamma_m$ is computed as followed
	\begin{align}
	\gamma_m = \frac{D\sum_{k \in \mathcal{K}_m }  a_{km}}{B^{\mathrm{bk}}_m}. \label{eqn:no_mec_edge_cloud_lat}
	\end{align}
	The uplink latency of users associated with edge node $m$ is
	\begin{align}
	T^{\mathrm{u}}_m  & = \max_{k \in \mathcal{K}_m}  {\Big\{ T^{\mathrm{u,\rm{fr}}}_m \Big \}} + \gamma_m.
	\end{align}	 

	For users associated with the cloud node, let $W^u$ denote the uplink communication capacity of the cloud node.  {Thus,} the uplink latency of these users is the longest latency of a given user
	\begin{align}
	T^{\mathrm{u}}_0 & = \max_{k \in \mathcal{K}_0} \Bigg\{D\frac{a_{k0} }{r_{k0} } \Bigg \}, \forall k \in \mathcal{K}_0,~\textrm{where}~r_{k0}  \leq W^u.
	\end{align}
	Hence, the total uplink latency of the whole network is
	\begin{align}
	T^{\mathrm{u}} & = \max \Bigg\{ \max_m \Big\{  {T^{\mathrm{u}}_m} \Big\},  {T^{\mathrm{u}}_0} \Bigg\} .
	\end{align}
 {	\begin{remark}
		Selecting the primal method or the primal-dual method does not impact the  amount of data sent by each user per iteration. The reason is that instead of sending $\mathbf{w}_k^t \in \mathbb{R}^d$, each user sends $\Delta \mathbf{v}_k^t \in \mathbb{R}^d$ having the same size $d$. Thus, in the scenarios where the system chooses primal method in initial iterations and primal-dual method later as in \cite{Shwartz2013} the computation of  the total uplink latency remains unchanged.
	\end{remark}}

	\subsection{Local Processing Model}
	Let $c_k$ (cycles/sample) be the number of processing cycles
	of MU $k$ to execute one sample of data, which assumes to be
	measured offline and known a prior \cite{TranINFOCOM2019}. Denoting the central
	processing unit (CPU) frequency of MU $k$ by $f_k$ (cycles/s), the
	computation time for the local training update at MU $k$ over $L$ local iterations is given by
	\begin{align}
	t_{k}^{\rm{cp}} = L_k \frac{c_k n_k}{f_k}.
	\end{align}
	Here, $L_k$ depends on the number of training passes that each client makes over its local dataset on each round, number of local data samples, and the local minibatch size \cite{McMahan2017}. Note that $t_{k}^{\rm{cp}}$ can be estimated by each MU before joining FL. Thus, the MEC knows $t_{k}^{\rm{cp}}$ as a prior. Since data generated by mobile users {usually follow the} power law \cite{Muchnik2013,BildTIT2015}, we assume that $t_{k}^{\rm{cp}}$  {also follows the} power law. Let $t_{k}^{\rm{cp}} $ be  {lower bounded and upper bounded by $t_{\min}^{\rm{cp}}$ and $ t_{\max}^{\rm{cp}}$, respectively.} The probability density of computing time is
		\begin{align}
			p(t^{\rm{cp}}) = \frac{\beta - 1}{t_{\min}^{\rm{cp}}} \Bigg( \frac{t^{\rm{cp}}}{t_{\min}^{\rm{cp}}} \Bigg)^{-\beta}.
		\end{align}
		With several real datasets \cite{BildTIT2015}, it is observed that $\beta  \in [1.47, 2.46]$. Here, we set $\beta$ at $1.6$, $t_{\min}^{\rm{cp}}$ at $0.2$. We also limit the maximum value of $t^{\rm{cp}}$ at $80$s.

		\section{In-Network Computation User Scheduling} \label{sec:user_shed}
		To minimize the  {training time} of the FL systems, we aim to minimize the latency that includes both up/down-link communication and computing latency for each FL iteration. Specifically, the latency per iteration depends on how we schedule the MUs to upload their updates, how the local model/updates from MUs are routed to the server, and how allocate communications resources for each up-link.  {Let $s \in \mathbb{S}$ denote a user scheduling scheme and $\mathbb{S}$ denote the set of all feasible scheduling schemes. The training time minimization problem can be formulated as follows.
		\begin{align}
			\min_{s,\mathbf{A},\mathbf{R}} T(s,\mathbf{A},\mathbf{R}),
		\end{align}
		where $T(s,\mathbf{A},\mathbf{R})$ is the training time of one iteration.} 
		\begin{figure}[t]
			\centering
			{\includegraphics[width=\linewidth]{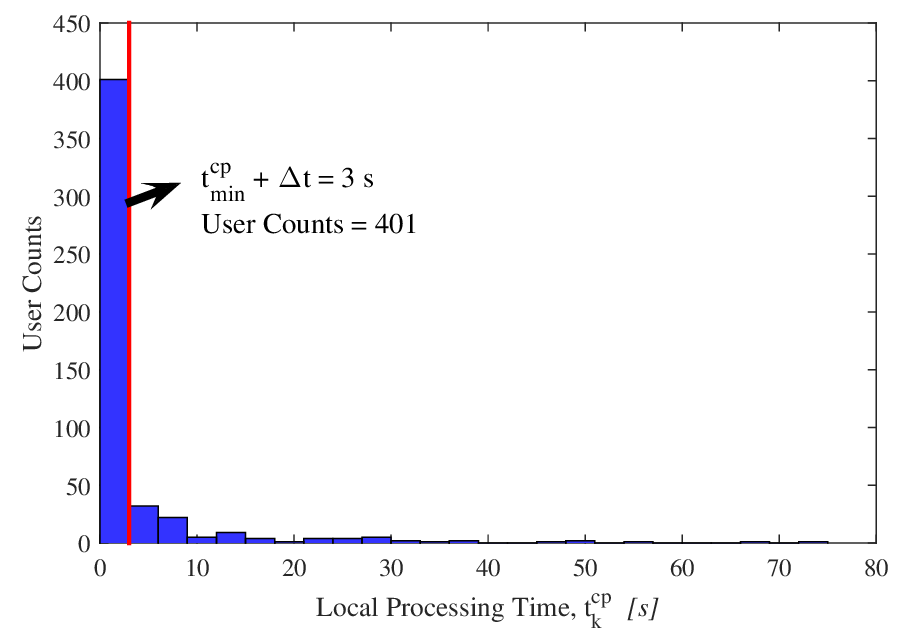}}
			\caption{An example of the distribution of $t_{k}^{\rm{cp}}$, where $K = 500$, $t_{\min}^{\rm{cp}} = 0.2$s, $t_{\max}^{\rm{cp}}  = 80$s and $\beta = 1.6$. \label{fig:power_law} }
		\end{figure}		
		For the widely accepted conventional scheduling scheme, denoted by $s_0$, the network operator  {begins} aggregating users' models after all users  {complete} their local update step \cite{TranINFOCOM2019}.  {Its training time} for one FL iteration is computed as follows :
		\begin{align}
			T(s_0,\mathbf{A},\mathbf{R}) =  T^{\mathrm{d}} + \max_k \{t_{k}^{\rm{cp}}  \} + T^{\mathrm{u}}(\mathbf{A},\mathbf{R}).
		\end{align} 
		This scheduling may suffer serious straggler effect  {due to} slow MUs. We observe that  {with different distributions of $t_{k}^{\rm{cp}}$, we can use different suitable user scheduling schemes. Let us consider two extreme cases where $t_{k}^{\rm{cp}}$ is either densely or dispersedly distributed. We apply two corresponding scheduling schemes with the upper bounds of the training time.}
		\begin{remark}
			If the local computing time of users is  {densely distributed}, i.e., $ t_{\max}^{\rm{cp}} - t_{\min}^{\rm{cp}} \leq \epsilon_0 T^{\mathrm{u}}(\mathbf{A},\mathbf{R}), \epsilon_0 \in (0,1) $, the network operator could wait for all users finishing updating local model, i.e., the $s_0$ scheduling scheme. The  {training time in one iteration} is then bounded by
			\begin{align}
			T(s_0,\mathbf{A},\mathbf{R})   \leq T^{\mathrm{d}} + t_{\min}^{\rm{cp}} + (1 + \epsilon_0) T^{\mathrm{u}}(\mathbf{A},\mathbf{R}).
			\end{align}
		\end{remark}
		
		We observe that if the difference of local processing time between the slowest and the fastest client is not significant, i.e., $ t_{\max}^{\rm{cp}} - t_{\min}^{\rm{cp}} \leq \epsilon_0 T^{\mathrm{u}}(\mathbf{A},\mathbf{R}), \epsilon_0 \in (0,1) $, the user scheduling does not highly impact the training time in one iteration. For example, given $100$ users and the aggregation time of collecting all models, $T^{\mathrm{u}}$, is $100$ seconds, if $80$ users  finish local processing at the same time after $79$ seconds and $20$ users finish local processing at the same time after $80$ seconds, 
		the training time in one iteration cannot be sooner than $T^{\mathrm{d}} + 79 + T^{\mathrm{u}}(\mathbf{A},\mathbf{R})$, but can not later than $T^{\mathrm{d}} + 80 + T^{\mathrm{u}}(\mathbf{A},\mathbf{R})$. Since the upper-bound of the difference of aggregation time between the earliest case and the slowest case is $1$ second which is very small compared with the aggregation time. In this case, it does not impact much on the training time.
		\begin{remark}
			Let us rank the local processing time $t_{k}^{\rm{cp}}$ in an increasing order. If the local processing time of users is  {dispersedly distributed}, i.e., $\min |t_{q}^{\rm{cp}} - t_{q+1}^{\rm{cp}} | \geq \frac{D}{W^u} , \forall q$, where $q$ is the index of the rank, the centralized/cloud server can collect users' models one by one. The  {training time in one iteration} is then
			\begin{align}
				T(s_1,\mathbf{A},\mathbf{R}) = T^{\mathrm{d}} + t_{\max}^{\rm{cp}} + \frac{D}{W^u}.
			\end{align}
		\end{remark}

		 {In the sequel,} we focus on addressing the scenario where $ t_{\max}^{\rm{cp}} - t_{\min}^{\rm{cp}} > \epsilon_0 T^{\mathrm{u}}(\mathbf{A},\mathbf{R}) $ and $\exists k,k', |t_{k}^{\rm{cp}} - t_{k'}^{\rm{cp}} | < \frac{D}{W^u}  $. As an illustrative example, Fig. \ref{fig:power_law} shows a distribution of $t_{k}^{\rm{cp}}$ of $500$ users following  {the power} law distribution, where $t_{\min}^{\rm{cp}}$ and $t_{\max}^{\rm{cp}}$ are $0.2$s and $80$s, respectively. As  can be  observed,  {within a small amount of time $\Delta t \ll t_{\max}^{\rm{cp}}$} during $[t_{\min}^{\rm{cp}}, t_{\min}^{\rm{cp}} + \Delta t]$, where $t_{\min}^{\rm{cp}} + \Delta t = 3$s, most users ($401$ out of $500$) finish their local update. Meanwhile, for the ResNet152's model size $D = 232$MB \cite{KerasPretrain}, the minimum uplink latency if all users directly send their model to the cloud via its uplink wireless channel is $T^{\mathrm{u}} = \frac{KD}{W^u} = 464$s. Since $t_{\min}^{\rm{cp}} + \Delta t = 3$s is much smaller than $T^{\mathrm{u}}$, while $t_{\max}^{\rm{cp}}$ is comparable with $T^{\mathrm{u}}$, it significantly stagnates the FL system if the server only starts the model aggregation after the last user completes its local update, i.e., using the scheduling $s_0$. Given this observation, we propose an in-network user scheduling mechanism to mitigate the straggler effect in FL systems in the next section. 
		
		\subsection{Bipartition User Scheduling Scheme}
		Our proposed bipartition user scheduling is summarized in Algorithm \ref{algorithm:user_scheduling}. Specifically, let $ {\mathcal{P}_1^t} = \{k \in \mathcal{K}^t | t_{k}^{\rm{cp}} \in  [t_{\min}^{\rm{cp}}, t_{\min}^{\rm{cp}} + \Delta t] \}$ be the user partition where $t_{k}^{\rm{cp}} \in  [t_{\min}^{\rm{cp}}, t_{\min}^{\rm{cp}} + \Delta t]$ and $ {\mathcal{P}_2^t} = \mathcal{K}^t \setminus \mathcal{P}_1^t$ denote the user partition of the rest users,  {at iteration $t$}. Let $T^{\mathrm{u}}(\mathbf{A}_{\mathcal{P}_j^t },\mathbf{R}_{\mathcal{P}_j^t },\mathcal{P}_j^t )$ denote the uplink latency which is returned by a network routing algorithm by aggregating of all users' models of $\mathcal{P}_j^t$ under the corresponding uplink routing and bandwidth allocation matrices $\mathbf{A}_{\mathcal{P}_j^t },\mathbf{R}_{\mathcal{P}_j^t}$, for $j \in \{1,2\}$.  {At each iteration $t$}, we first begin aggregating
		the global model of $\mathcal{P}_1^t$ at time $t_{\min}^{\rm{cp}} + \Delta t$. The time required for these users in $\mathcal{P}_1^t$ to complete the aggregation process is $t_{\mathcal{P}_1^t} = T^{\mathrm{d}} + t_{\min}^{\rm{cp}} + \Delta t + T^{\mathrm{u}}(\mathbf{A}_{\mathcal{P}_1^t},\mathbf{R}_{\mathcal{P}_1^t},\mathcal{P}_1^t)$. If $t_{\mathcal{P}_1^t} > T^{\mathrm{d}} + t_{\max}^{\rm{cp}}$, we aggregate $\mathcal{P}_2^t$. Otherwise, the rest users will wait until the slowest user finishing local processing, i.e., $T^{\mathrm{d}} + t_{\max}^{\rm{cp}}$.  
		
		\begin{algorithm}    [t]                
			\caption{Bipartition User Scheduling Scheme}          
			\label{algorithm:user_scheduling}                           
			{
				\begin{algorithmic}[1]                    
					\Require $\mathcal{K}^t$, $\{t_{k}^{\rm{cp}}, \forall k \in \mathcal{K}^t\}$, $\Delta t$ and a network routing algorithm.
					\Ensure $T'$
					\State Partition users into $\mathcal{P}_1 = \{k \in \mathcal{K}^t | t_{k}^{\rm{cp}} \in  [t_{\min}^{\rm{cp}} + \Delta t] \}$ and $\mathcal{P}_2 = \mathcal{K}^t \setminus \mathcal{P}_1$.
					\State Begin aggregating $\mathcal{P}_1$ at time $t_{\min}^{\rm{cp}} + \Delta t$. Then, compute $t_{\mathcal{P}_1} = t_{\min}^{\rm{cp}} + \Delta t + T^{\mathrm{u}}(\mathbf{A}_{\mathcal{P}_1},\mathbf{R}_{\mathcal{P}_1},\mathcal{P}_1)$. 
					\If  {$t_{\mathcal{P}_1} >T^{\mathrm{d}} + t_{\max}^{\rm{cp}}$}
					\State Begin aggregating $\mathcal{P}_2$ at the time $t_{\mathcal{P}_1}$.
					\Else
					\State Begin aggregating $\mathcal{P}_2$ at the time $T^{\mathrm{d}} + t_{\max}^{\rm{cp}}$.
					\EndIf						
				\end{algorithmic}}
			\end{algorithm}
			
			\subsection{System's Latency Analysis}
			We now analyze the  {training time} of the FL system in one iteration resulting from the proposed user scheduling scheme. This bipartition scheduling scheme's whole network latency of one FL iteration is computed by
	\begin{figure*}
	\begin{equation}
	\{ \phi_k^t[0],\bm{\phi}_k^t[1]  \} = \begin{cases}
	\{ n_k, 	\mathbf{w}_k^t \},& \text{if we choose the primal method}, \\
	\{ 1 , \Delta \mathbf{v}_k^t\}, & \text{if we choose the primal-dual method}. 
	\end{cases} \label{eqn:user_packet}
	\end{equation}
	\end{figure*}			
			\begin{align}
			 {T(s_b,\mathbf{A},\mathbf{R})} = T^{\mathrm{d}} + \max \{t_{\mathcal{P}_1^t}, t_{\max}^{\rm{cp}}\} +T^{\mathrm{u}}(\mathbf{A}_{\mathcal{P}_2^t},\mathbf{R}_{\mathcal{P}_2^t},\mathcal{P}_2^t),
			\end{align}
			 {where $s_b$ is the proposed bipartition user scheduling scheme.} 
			\begin{theorem}
				If there exist $\epsilon_1, \epsilon_2, \epsilon_3 \in (0,1)$ and $\epsilon_0 > \epsilon_2+ \epsilon_3$ that $|\mathcal{P}_2|  \leq \epsilon_1 K$, $T^{\mathrm{u}}(\mathbf{A}_{\mathcal{P}_2^t},\mathbf{R}_{\mathcal{P}_2^t},\mathcal{P}_2^t) \leq \epsilon_2 T^{\mathrm{u}}(\mathbf{A},\mathbf{R})$ and $\Delta t  \leq \epsilon_3 T^{\mathrm{u}}(\mathbf{A},\mathbf{R})$, $T(s_b,\mathbf{A},\mathbf{R})$ is upper bounded by
				\begin{align}
			\nonumber		 {T(s_b,\mathbf{A},\mathbf{R})} \leq & T^{\mathrm{d}}+ \max \Big\{  t_{\min}^{\rm{cp}}  + (1 +\epsilon_2 +\epsilon_3 ) T^{\mathrm{u}}(\mathbf{A},\mathbf{R}), \\
					 & t_{\max}^{\rm{cp}} + \epsilon_2 T^{\mathrm{u}}(\mathbf{A},\mathbf{R})\Big\}   <  {T(s_0,\mathbf{A},\mathbf{R})}.
				\end{align}
			\end{theorem}
			\begin{proof}
				If $t_{\mathcal{P}_1^t} > T^{\mathrm{d}} + t_{\max}^{\rm{cp}}$, 
				\begin{align}
					\nonumber	 {T(s_b,\mathbf{A},\mathbf{R})} & =  T^{\mathrm{d}} + t_{\min}^{\rm{cp}} + \Delta t + T^{\mathrm{u}} (\mathbf{A}_{\mathcal{P}_1^t},\mathbf{R}_{\mathcal{P}_1^t},\mathcal{P}_1^t) \\
	\nonumber				& ~~~+ T^{\mathrm{u}} (\mathbf{A}_{\mathcal{P}_2^t},\mathbf{R}_{\mathcal{P}_2^t},\mathcal{P}_2^t)\\
	\nonumber				&  \leq T^{\mathrm{d}} + t_{\min}^{\rm{cp}}  + (1 +\epsilon_2  +\epsilon_3) T^{\mathrm{u}} (\mathbf{A},\mathbf{R}) \\
	& < T^{\mathrm{d}}+ t_{\max}^{\rm{cp}} + T^{\mathrm{u}} (\mathbf{A},\mathbf{R}) =   {T(s_0,\mathbf{A},\mathbf{R})}. \label{eqn:temp_a}
				\end{align}
				Otherwise, 
				\begin{align}
				\nonumber	 {T(s_b,\mathbf{A},\mathbf{R})} & = T^{\mathrm{d}} + t_{\max}^{\rm{cp}} + T^{\mathrm{u}} (\mathbf{A}_{\mathcal{P}_2^t},\mathbf{R}_{\mathcal{P}_2^t},\mathcal{P}_2^t) \\
\nonumber				& \leq T^{\mathrm{d}}+ t_{\max}^{\rm{cp}} + \epsilon_2  T^{\mathrm{u}} (\mathbf{A},\mathbf{R}) \\
					&  < T^{\mathrm{d}} + t_{\max}^{\rm{cp}} + T^{\mathrm{u}} (\mathbf{A},\mathbf{R}) =   {T(s_0,\mathbf{A},\mathbf{R})}. \label{eqn:temp_b}
				\end{align}		
				From (\ref{eqn:temp_a}) and (\ref{eqn:temp_b}), the upper bound of the whole network latency of proposed user scheduling $T'$ is achieved which is always smaller than that of the conventional user scheduling $T$.
			\end{proof}

{Let us study a simple example where there is a star network consisting of $K = 500$ users and a single cloud node with $W^u = W^d = 2$Gbps. We assume that ResNet152's model is considered. Reusing the distribution in Fig. \ref{fig:power_law}, the aggregation time without the proposed user scheduling is $T(s_0,\mathbf{A},\mathbf{R}) = D/W^d + t_{\max}^{\rm{cp}} + KD/W^u = 0.928 + 80 + 464 = 544.928$s. With the proposed user scheduling, $t_{\min}^{\rm{cp}} + \Delta t = 3$s and $|\mathcal{P}_1| = 401$ users, and $t_{\mathcal{P}_1} = 0.928 + 3 + 372.128 = 376.056$s. Since $t_{\mathcal{P}_1} > T^{\mathrm{d}} + t_{\max}^{\rm{cp}} = 0.928 + 80 = 80.928$s,  the aggregation time is $T(s_b,\mathbf{A},\mathbf{R}) = T^{\mathrm{d}} + t_{\min}^{\rm{cp}} + \Delta t + T^{\mathrm{u}} (\mathbf{A}_{\mathcal{P}_1^t},\mathbf{R}_{\mathcal{P}_1^t},\mathcal{P}_1^t) + T^{\mathrm{u}} (\mathbf{A}_{\mathcal{P}_2^t},\mathbf{R}_{\mathcal{P}_2^t},\mathcal{P}_2^t) = t_{\mathcal{P}_1} + (K - |\mathcal{P}_1|)D/W^u = 376.056 + 91.872 = 467.928$s. Thus, $T(s_b,\mathbf{A},\mathbf{R}) < T(s_0,\mathbf{A},\mathbf{R})$. In another case, when $K = 50$ and $|\mathcal{P}_1| = 40$ users, $T(s_0,\mathbf{A},\mathbf{R}) = D/W^d + t_{\max}^{\rm{cp}} + KD/W^u = 0.928 + 80 + 46.4 = 127.328$s and $t_{\mathcal{P}_1} = 0.928 + 3 + 37.2128 = 41.1408$s. Here $t_{\mathcal{P}_1} < T^{\mathrm{d}} + t_{\max}^{\rm{cp}}$,  the aggregation time is thus $T(s_b,\mathbf{A},\mathbf{R}) = T^{\mathrm{d}} + t_{\max}^{\rm{cp}} + T^{\mathrm{u}} (\mathbf{A}_{\mathcal{P}_2^t},\mathbf{R}_{\mathcal{P}_2^t},\mathcal{P}_2^t) = 0.928 + 80 + 9.28 = 90.208$s. It is also smaller than $T(s_0,\mathbf{A},\mathbf{R})$.}
	\section{In-Network Aggregation Design}	\label{sec:ina}	
	
	We now introduce the in-network  {aggregation design} that allows edge nodes to support the server for addictive weighting users' local models.  {Back to early $2000$s, the concept of in-network computation was well-studied for wireless sensor networks (WSNs), e.g., \cite{FasoloWC2007} due to sensors' limited communications, computing, storage capabilities. The core idea of in-network computation is to design data structures to better represent the information collected/generated at each sensor for each specific application \cite{FasoloWC2007}. Analogously, under FL, mobile users also do not transmit their raw data to the server. We then can interpret their local models as their data representation. To obtain the global model, defined in (\ref{eqn:global_model}), we first design the user packet which plays the role of data representation in a in-network computation solution.} This packet design  can be tailored to adapt  {with} two FL schemes. Specifically, let $\bm{\phi}_k^{t} = \{\phi_k^t[0],\bm{\phi}_k^t[1]  \}$ denote the local message of users $k$ at iteration $t$ which are generated by following (\ref{eqn:user_packet}).
	\subsection{Aggregation Function}
	In this section, we propose a general aggregation function which is suitable for  {the two methods solving FL problems}. If we use  {the} primal method, we observe that
	\begin{align}
	n  = \sum_k n_k = \sum_k \phi_k^t[0].
	\end{align}
	Similarly, for {the} primal-dual method, we  also observe that
	\begin{align}
	|\mathcal{K}^t| &  = \sum_k \phi_k^t[0].
	\end{align}	
	Hence, let $z$ be a hyperparameter where $z = 1$ if we use the primal-dual method and $z=0$ otherwise.  {To preserve the return of (\ref{eqn:global_model}), the global model of the cloud at each iteration can be computed from the users' packets designed in (\ref{eqn:user_packet}) as follows}
	\begin{align}
	\bm{\psi}^{t+1} = z\bm{\psi}^{t} + \frac{\sum_k \phi_k^t[0] \bm{\phi}_k^t[1]}{\sum_k \phi_k^t[0]}. \label{eqn:agg_function}
	\end{align}
	Thus, we now can use (\ref{eqn:agg_function}) as the aggregation function generalized for both primal and primal-dual methods solving FL. 
	
		\subsection{In-network Aggregation Process with Bipartition User Scheduling Scheme}

				\begin{figure}
				\begin{tikzpicture}[x=0.47pt,y=0.47pt,yscale=-1,xscale=1]
				
				\draw  [color={rgb, 255:red, 255; green, 255; blue, 255 }  ,draw opacity=1 ][fill={rgb, 255:red, 126; green, 211; blue, 33 }  ,fill opacity=1 ] (198.18,77.38) .. controls (198.18,72.31) and (202.67,68.2) .. (208.21,68.2) .. controls (213.76,68.2) and (218.25,72.31) .. (218.25,77.38) .. controls (218.25,82.44) and (213.76,86.55) .. (208.21,86.55) .. controls (202.67,86.55) and (198.18,82.44) .. (198.18,77.38) -- cycle ;
				\draw  [color={rgb, 255:red, 255; green, 255; blue, 255 }  ,draw opacity=1 ][fill={rgb, 255:red, 74; green, 144; blue, 226 }  ,fill opacity=1 ] (126.61,209.84) .. controls (126.61,204.77) and (131.11,200.66) .. (136.65,200.66) .. controls (142.19,200.66) and (146.69,204.77) .. (146.69,209.84) .. controls (146.69,214.91) and (142.19,219.02) .. (136.65,219.02) .. controls (131.11,219.02) and (126.61,214.91) .. (126.61,209.84) -- cycle ;
				\draw  [color={rgb, 255:red, 255; green, 255; blue, 255 }  ,draw opacity=1 ][fill={rgb, 255:red, 74; green, 144; blue, 226 }  ,fill opacity=1 ] (193.81,210.64) .. controls (193.81,205.57) and (198.31,201.46) .. (203.85,201.46) .. controls (209.39,201.46) and (213.89,205.57) .. (213.89,210.64) .. controls (213.89,215.71) and (209.39,219.82) .. (203.85,219.82) .. controls (198.31,219.82) and (193.81,215.71) .. (193.81,210.64) -- cycle ;
				\draw  [color={rgb, 255:red, 255; green, 255; blue, 255 }  ,draw opacity=1 ][fill={rgb, 255:red, 74; green, 144; blue, 226 }  ,fill opacity=1 ] (262.76,211.44) .. controls (262.76,206.37) and (267.25,202.26) .. (272.8,202.26) .. controls (278.34,202.26) and (282.83,206.37) .. (282.83,211.44) .. controls (282.83,216.51) and (278.34,220.61) .. (272.8,220.61) .. controls (267.25,220.61) and (262.76,216.51) .. (262.76,211.44) -- cycle ;
				\draw  [color={rgb, 255:red, 255; green, 255; blue, 255 }  ,draw opacity=1 ][fill={rgb, 255:red, 74; green, 144; blue, 226 }  ,fill opacity=1 ] (164.61,118.53) -- (184.82,100.52) -- (184.35,126.32) -- (178.39,123.97) -- (143.75,197.36) -- (135.92,194.27) -- (170.56,120.88) -- cycle ;
				\draw  [color={rgb, 255:red, 255; green, 255; blue, 255 }  ,draw opacity=1 ][fill={rgb, 255:red, 74; green, 144; blue, 226 }  ,fill opacity=1 ] (195.06,121.67) -- (206.31,98.53) -- (216.99,121.89) -- (210.38,121.82) -- (209.47,197.48) -- (200.77,197.39) -- (201.68,121.74) -- cycle ;
				\draw  [color={rgb, 255:red, 255; green, 255; blue, 255 }  ,draw opacity=1 ][fill={rgb, 255:red, 74; green, 144; blue, 226 }  ,fill opacity=1 ] (226.85,125.93) -- (228.07,100.25) -- (247.39,118.92) -- (241.19,121.04) -- (270.65,193.18) -- (262.5,195.96) -- (233.05,123.82) -- cycle ;
				\draw  [color={rgb, 255:red, 255; green, 255; blue, 255 }  ,draw opacity=1 ][fill={rgb, 255:red, 126; green, 211; blue, 33 }  ,fill opacity=1 ] (411.12,77.81) .. controls (411.12,72.74) and (415.62,68.63) .. (421.16,68.63) .. controls (426.7,68.63) and (431.2,72.74) .. (431.2,77.81) .. controls (431.2,82.88) and (426.7,86.99) .. (421.16,86.99) .. controls (415.62,86.99) and (411.12,82.88) .. (411.12,77.81) -- cycle ;
				\draw  [color={rgb, 255:red, 255; green, 255; blue, 255 }  ,draw opacity=1 ][fill={rgb, 255:red, 74; green, 144; blue, 226 }  ,fill opacity=1 ] (339.56,210.27) .. controls (339.56,205.21) and (344.05,201.1) .. (349.6,201.1) .. controls (355.14,201.1) and (359.63,205.21) .. (359.63,210.27) .. controls (359.63,215.34) and (355.14,219.45) .. (349.6,219.45) .. controls (344.05,219.45) and (339.56,215.34) .. (339.56,210.27) -- cycle ;
				\draw  [color={rgb, 255:red, 255; green, 255; blue, 255 }  ,draw opacity=1 ][fill={rgb, 255:red, 74; green, 144; blue, 226 }  ,fill opacity=1 ] (406.76,211.07) .. controls (406.76,206) and (411.25,201.9) .. (416.8,201.9) .. controls (422.34,201.9) and (426.83,206) .. (426.83,211.07) .. controls (426.83,216.14) and (422.34,220.25) .. (416.8,220.25) .. controls (411.25,220.25) and (406.76,216.14) .. (406.76,211.07) -- cycle ;
				\draw  [color={rgb, 255:red, 255; green, 255; blue, 255 }  ,draw opacity=1 ][fill={rgb, 255:red, 74; green, 144; blue, 226 }  ,fill opacity=1 ] (475.71,210.27) .. controls (475.71,205.21) and (480.2,201.1) .. (485.74,201.1) .. controls (491.29,201.1) and (495.78,205.21) .. (495.78,210.27) .. controls (495.78,215.34) and (491.29,219.45) .. (485.74,219.45) .. controls (480.2,219.45) and (475.71,215.34) .. (475.71,210.27) -- cycle ;
				\draw  [color={rgb, 255:red, 255; green, 255; blue, 255 }  ,draw opacity=1 ][fill={rgb, 255:red, 74; green, 144; blue, 226 }  ,fill opacity=1 ] (362.84,158.33) -- (377.61,151.83) -- (382.58,166.12) -- (376.62,163.77) -- (360.19,198.59) -- (352.36,195.5) -- (368.79,160.68) -- cycle ;
				\draw  [color={rgb, 255:red, 255; green, 255; blue, 255 }  ,draw opacity=1 ][fill={rgb, 255:red, 74; green, 144; blue, 226 }  ,fill opacity=1 ] (427.94,159.01) -- (443.03,152.44) -- (448.05,167.01) -- (441.99,164.6) -- (425.6,199.01) -- (417.62,195.84) -- (434.01,161.43) -- cycle ;
				\draw  [color={rgb, 255:red, 255; green, 255; blue, 255 }  ,draw opacity=1 ][fill={rgb, 255:red, 74; green, 144; blue, 226 }  ,fill opacity=1 ] (455.27,164.28) -- (461.25,150.25) -- (475.82,157.27) -- (469.62,159.38) -- (483.6,193.62) -- (475.45,196.4) -- (461.47,162.16) -- cycle ;
				\draw  [color={rgb, 255:red, 255; green, 255; blue, 255 }  ,draw opacity=1 ][fill={rgb, 255:red, 208; green, 2; blue, 27 }  ,fill opacity=1 ] (373.6,136.06) .. controls (373.6,130.99) and (378.09,126.89) .. (383.63,126.89) .. controls (389.18,126.89) and (393.67,130.99) .. (393.67,136.06) .. controls (393.67,141.13) and (389.18,145.24) .. (383.63,145.24) .. controls (378.09,145.24) and (373.6,141.13) .. (373.6,136.06) -- cycle ;
				\draw  [color={rgb, 255:red, 255; green, 255; blue, 255 }  ,draw opacity=1 ][fill={rgb, 255:red, 208; green, 2; blue, 27 }  ,fill opacity=1 ] (439.05,136.06) .. controls (439.05,130.99) and (443.55,126.89) .. (449.09,126.89) .. controls (454.63,126.89) and (459.12,130.99) .. (459.12,136.06) .. controls (459.12,141.13) and (454.63,145.24) .. (449.09,145.24) .. controls (443.55,145.24) and (439.05,141.13) .. (439.05,136.06) -- cycle ;
				\draw  [color={rgb, 255:red, 255; green, 255; blue, 255 }  ,draw opacity=1 ][fill={rgb, 255:red, 74; green, 144; blue, 226 }  ,fill opacity=1 ] (399.35,93.76) -- (411.1,90.49) -- (414.33,101.32) -- (409.81,99.04) -- (395.54,122.66) -- (389.6,119.66) -- (403.87,96.04) -- cycle ;
				\draw  [color={rgb, 255:red, 255; green, 255; blue, 255 }  ,draw opacity=1 ][fill={rgb, 255:red, 74; green, 144; blue, 226 }  ,fill opacity=1 ] (424.25,99.53) -- (429.14,89.64) -- (439.93,94.18) -- (435.2,95.8) -- (444.79,119.29) -- (438.57,121.41) -- (428.98,97.92) -- cycle ;
				\draw  [color={rgb, 255:red, 255; green, 255; blue, 255 }  ,draw opacity=1 ][fill={rgb, 255:red, 126; green, 211; blue, 33 }  ,fill opacity=1 ] (542.03,80.2) .. controls (542.03,75.14) and (546.53,71.03) .. (552.07,71.03) .. controls (557.61,71.03) and (562.11,75.14) .. (562.11,80.2) .. controls (562.11,85.27) and (557.61,89.38) .. (552.07,89.38) .. controls (546.53,89.38) and (542.03,85.27) .. (542.03,80.2) -- cycle ;
				\draw  [color={rgb, 255:red, 255; green, 255; blue, 255 }  ,draw opacity=1 ][fill={rgb, 255:red, 208; green, 2; blue, 27 }  ,fill opacity=1 ] (540.29,137.66) .. controls (540.29,132.59) and (544.78,128.48) .. (550.32,128.48) .. controls (555.87,128.48) and (560.36,132.59) .. (560.36,137.66) .. controls (560.36,142.73) and (555.87,146.84) .. (550.32,146.84) .. controls (544.78,146.84) and (540.29,142.73) .. (540.29,137.66) -- cycle ;
				\draw  [color={rgb, 255:red, 255; green, 255; blue, 255 }  ,draw opacity=1 ][fill={rgb, 255:red, 74; green, 144; blue, 226 }  ,fill opacity=1 ] (541.16,201.5) .. controls (541.16,196.43) and (545.65,192.32) .. (551.2,192.32) .. controls (556.74,192.32) and (561.23,196.43) .. (561.23,201.5) .. controls (561.23,206.57) and (556.74,210.67) .. (551.2,210.67) .. controls (545.65,210.67) and (541.16,206.57) .. (541.16,201.5) -- cycle ;
				
				\draw (108.44,155.59) node [anchor=north west][inner sep=0.75pt]    {$\bm{\phi} _{1}^{t}$};
				\draw (171.28,167.56) node [anchor=north west][inner sep=0.75pt]    {$\bm{\phi} _{2}^{t}$};
				\draw (265.41,158.78) node [anchor=north west][inner sep=0.75pt]    {$\bm{\phi} _{K}^{t}$};
				\draw (192.79,225.69) node [anchor=north west][inner sep=0.75pt]   [align=left] {{\fontfamily{ptm}\selectfont (a)}};
				\draw (331.86,156.82) node [anchor=north west][inner sep=0.75pt]    {$\bm{\phi} _{1}^{t}$};
				\draw (397.32,163.2) node [anchor=north west][inner sep=0.75pt]    {$\bm{\phi} _{2}^{t}$};
				\draw (478.35,159.21) node [anchor=north west][inner sep=0.75pt]    {$\bm{\phi} _{K}^{t}$};
				\draw (405.67,224.12) node [anchor=north west][inner sep=0.75pt]   [align=left] {{\fontfamily{ptm}\selectfont (b)}};
				\draw (371.94,86.96) node [anchor=north west][inner sep=0.75pt]    {$\bm{\varphi} _{1}^{t}$};
				\draw (444.12,86.96) node [anchor=north west][inner sep=0.75pt]    {$\bm{\varphi} _{M}^{t}$};
				\draw (573.14,68.11) node [anchor=north west][inner sep=0.75pt]   [align=left] {{\fontfamily{helvet}\selectfont Cloud}};
				\draw (575.14,126.36) node [anchor=north west][inner sep=0.75pt]   [align=left] {{\fontfamily{helvet}\selectfont Edge}};
				\draw (574.89,190.2) node [anchor=north west][inner sep=0.75pt]   [align=left] {Users};
				
				\end{tikzpicture}

			\caption{The logical view of (a) conventional network model and (b) multi-tier edge network model with INA.   \label{fig::logic_net} }	
			\end{figure}
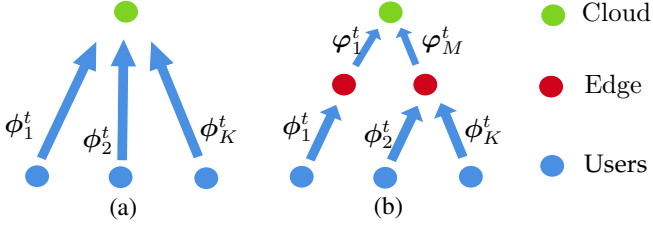
				
		Under the proposed bipartition scheduling scheme $s_b$, users are scheduled to aggregate their models in two partitions. In this section, we propose a novel in-network aggregation (INA) process, as illustrated in Fig. \ref{fig::logic_net}, that enable edge nodes to support {the cloud} node in decentralizing the aggregation process. Let $\bm{\chi}^t_{m,j}$ denote the local aggregated  model of edge node $m$, for user partition $\mathcal{P}_j$, such that
		\begin{align}
		\bm{\chi}^t_{m,j}  = \frac{\sum_{k \in \mathcal{K}_m, k \in \mathcal{P}_j} \phi_k^t[0] \bm{\phi}_k^t[1]}{\sum_{k \in \mathcal{K}_m,k \in \mathcal{P}_j} \phi_k^t[0]}.
		\end{align}	
		Let $\bm{\varphi}_{m,j}^t = \{\varphi_{m,j}^t[0],\bm{\varphi}_{m,j}^t[1] \}$ denote the message edge node $m$ sends to  {the cloud} node, for user partition $\mathcal{P}_j$, such that
		\begin{align}
		\{ \varphi_{m,j}^t[0],\bm{\varphi}_{m,j}^t[1]  \} = \Big \{ \sum_{k \in \mathcal{K}_m,k \in \mathcal{P}_j} \phi_k^t[0] , \bm{\chi}^t_{m,j} \Big \}.
		\end{align}
		Let $\bm{\lambda}^t_j = \{\lambda^t_j[0],\bm{\lambda}^t_j[1] \}$ denote the weight and the parameters of the aggregated model for user partition $\mathcal{P}_j$. After all edge nodes' messages $\bm{\varphi}_{m,j}^t$ are sent to the cloud, we compute $\bm{\lambda}^t_j$ as follows:
		\begin{align}
\begin{cases}
 \lambda^t_j[0] & =  \sum_{k \in \mathcal{K}_0, k \in \mathcal{P}_j} \phi_k^t[0] + \sum_{m}\varphi_{m,j}^t[0], \\
 \bm{\lambda}^t_j[1] & =  \frac{\sum_{k \in \mathcal{K}_0, k \in \mathcal{P}_j} \phi_k^t[0] \bm{\phi}_k^t[1] + \sum_{m}\varphi_{m,j}^t[0]\bm{\varphi}_{m,j}^t[1]}{\sum_{k \in \mathcal{K}_0, k \in \mathcal{P}_j} \phi_k^t[0] + \sum_{m}\varphi_{m,j}^t[0]}.
\end{cases}		
		\end{align}
		\begin{proposition} \label{proposition:preserve_return}
			In order to preserve the return of (\ref{eqn:agg_function}), the global model is computed as follows,
			\begin{align}
			\bm{\psi}^{t+1}  = z\bm{\psi}^{t}  + \frac{\sum_j \lambda^t_j[0] \bm{\lambda}^t_j[1]}{\sum_j \lambda^t_j[0]}.  \label{eqn:cloud_process}
			\end{align}	
		\end{proposition}
		\begin{proof}
			We have:
			\begin{align}
	\nonumber		&	\sum_j \lambda^t_j[0] \\
\nonumber	& =\hspace{-0.2cm} \sum_{k \in \mathcal{K}_0, k \in \mathcal{P}_1} \hspace{-0.2cm} \phi_k^t[0] +  \sum_{m}\varphi_{m,1}^t[0] + \hspace{-0.2cm}\sum_{k \in \mathcal{K}_0, k \in \mathcal{P}_2} \hspace{-0.2cm} \phi_k^t[0] + \sum_{m}\varphi_{m,2}^t[0]\\
	\nonumber			& = \hspace{-0.2cm} \sum_{k \in \mathcal{K}_0, k \in \mathcal{P}_1} \hspace{-0.2cm} \phi_k^t[0] + \hspace{-0.35cm} \sum_{k \in \mathcal{K}_m, k \in \mathcal{P}_1 } \hspace{-0.4cm} \phi_k^t[0]+ \hspace{-0.2cm} \sum_{k \in \mathcal{K}_0, k \in \mathcal{P}_2} \hspace{-0.2cm} \phi_k^t[0]+ \hspace{-0.3cm} \sum_{k \in \mathcal{K}_m, k \in \mathcal{P}_2 } \hspace{-0.4cm} \phi_k^t[0]\\
	& = \sum_k \phi_k^t[0]. \label{eqn:tmp1}
			\end{align}
			We also have
			\begin{align}
				& \nonumber \sum_j \lambda^t_j[0] \bm{\lambda}^t_j[1]\\
			\nonumber	 & = \sum_j \Bigg ( \sum_{k \in \mathcal{K}_0, k \in \mathcal{P}_j} \phi_k^t[0] + \sum_{m}\varphi_{m,j}^t[0]\Bigg ) \\
\nonumber			& \hspace{1cm} \Bigg (\frac{\sum_{k \in \mathcal{K}_0, k \in \mathcal{P}_j} \phi_k^t[0] \bm{\phi}_k^t[1] + \sum_{m}\varphi_{m,j}^t[0]\bm{\varphi}_{m,j}^t[1]}{\sum_{k \in \mathcal{K}_0, k \in \mathcal{P}_j} \phi_k^t[0] + \sum_{m}\varphi_{m,j}^t[0]} \Bigg ) \\
			\nonumber	 & = \sum_j \Bigg ( \sum_{k \in \mathcal{K}_0, k \in \mathcal{P}_j} \phi_k^t[0] \bm{\phi}_k^t[1] + \sum_{m}\varphi_{m,j}^t[0]\bm{\varphi}_{m,j}^t[1] \Bigg )\\
			\nonumber	 & = \sum_j \Bigg ( \sum_{k \in \mathcal{K}_0, k \in \mathcal{P}_j} \phi_k^t[0] \bm{\phi}_k^t[1] \\
\nonumber	 & \hspace{1cm}			+ \sum_{m}\Bigg(\sum_{k \in \mathcal{K}_m,k \in \mathcal{P}_j} \phi_k^t[0] \frac{\sum_{k \in \mathcal{K}_m, k \in \mathcal{P}_j} \phi_k^t[0] \bm{\phi}_k^t[1]}{\sum_{k \in \mathcal{K}_m,k \in \mathcal{P}_j} \phi_k^t[0]} \Bigg) \Bigg )\\
		\nonumber		 & =\sum_j \Bigg ( \sum_{k \in \mathcal{K}_0, k \in \mathcal{P}_j} \phi_k^t[0] \bm{\phi}_k^t[1] + \sum_{m} \sum_{k \in \mathcal{K}_m, k \in \mathcal{P}_j} \phi_k^t[0] \bm{\phi}_k^t[1] \Bigg )
		\end{align}
		\begin{align}
				 = \sum_k \phi_k^t[0] \bm{\phi}_k^t[1]. \hspace{5.4cm}\label{eqn:tmp2}
			\end{align}
			From (\ref{eqn:tmp1}) and (\ref{eqn:tmp2}), the return of (\ref{eqn:agg_function}) is preserved when the global model is aggregated using (\ref{eqn:cloud_process}).
		\end{proof}
		\begin{theorem} (Theorem $2$ of \cite{WangJSAC2019} as well as Theorem $4.2$ and Theorem $4.3$ of \cite{Ma2017}) Assuming that $l_i(\cdot)$  {is} convex with $1/\mu$-smoothness and $r(\cdot)$ is 1-strongly convex. If the primal method is used, given that $\mathbf{w}^{*}$ is the global optimal model, after $\iota$ learning rounds, we have the following convergence upperbound:
			\begin{align}
				P(\mathbf{w}^{\iota}) - P(\mathbf{w}^{*}) \leq \theta_g, \label{eqn:FedAvg_convergence}
			\end{align}
			where $\theta_g$ is the global accuracy. Here, $\iota$ is upperbounded by $\mathcal{O}(1/\theta_g)$. On the other hand, if the primal-dual method is used, after $\iota'$ learning rounds, we have the following convergence upperbound:
			\begin{align}
			\mathbb{E} \Bigg [P(\mathbf{w}(\bm{\alpha}^{\iota'} )) - G(\bm{\alpha}^{\iota'} ) \Bigg]\leq \theta_g .\label{eqn:CoCoA_convergence}
			\end{align}
			 Here, $\iota'$ is upperbounded by $\frac{\mathcal{O}(\log(1/\theta_g))}{1-\theta_l}$, where $\theta_l$ is the local accuracy defining at Assumption $4.1$ of \cite{Ma2017}.
		\end{theorem}
		\begin{proof}
		From Proposition \ref{proposition:preserve_return}, it can be seen that the INA process always remains the returns of the FL aggregated model for both FedAvg and CoCoA. Thus, if FedAvg is used, Eq. (\ref{eqn:FedAvg_convergence}) holds by following Theorem 2 of \cite{WangJSAC2019}. Similarly, if CoCoA is used, Eq. (\ref{eqn:CoCoA_convergence}) holds by following Theorem $4.2$ and Theorem $4.3$ of \cite{Ma2017}.
		\end{proof}
		\begin{remark}
			In a special case where no user is directly associated with the cloud node, the proposed edge network architecture and the INA process could reduce the traffic and computing overhead at the cloud node by a factor of $K/M$ in comparison with conventional FL star network topology,  {where $K$ is the number of users and $M$ is the number of edge nodes.}
		\end{remark}
		\begin{proof}
			The traffic overhead at the cloud node is proportional to the number of models the cloud node received. For computing overhead, we also assume that the aggregation operation at the cloud runs linearly to the number of received models. Hence, under the conventional FL star network topology, the cloud node receives $K$ models from its users. 
			
			In contrast, with the proposed edge network architecture using the INA, when no user is directly associated with the cloud node, i.e., $\mathcal{K}_0 =  \emptyset $, the cloud node only receives $M$ aggregated models from the edge nodes. As a result, in practice where $K >> M$, the two mentioned overheads are reduced by a factor of $K/M$.
		\end{proof}

	Under the proposed INA process,  {the transmission latency $\gamma_m$ between an edge node $m$ and the cloud node is revised}. If there is no user associate with an edge node $m$, i.e., $\sum_{k}  a_{km} = 0$, $\gamma_m$ is zero. Otherwise, since edge node $m$ only needs to send the aggregated model the cloud node, $\gamma_m$ is  {revised} as follows
	\begin{align}
		\gamma_m = \min \Bigg \{\frac{D}{B^{\mathrm{bk}}_m}, \frac{D\sum_{k \in \mathcal{K}_m }  a_{km}}{B^{\mathrm{bk}}_m} \Bigg \}. \label{eqn:revised_edge_cloud_latency}
	\end{align}

	\section{Network Routing and Resource Allocation Framework for FL} \label{section:problem_for}
	Following the user scheduling scheme and the INA process proposed above, in this section, we  {aim} to optimize  {the} network' resource allocation and the routing matrices to minimize the  {uplink aggregation latency in one iteration} for a given partition $\mathcal{P}_j^t$.
	When the proposed INA process is considered, the joint routing and resource optimization problem for users in $\mathcal{P}_j^t$ can be written as follows:
	\begin{subequations}
			\begin{align}
			\nonumber \mathscr F_1^u: & \min_{\mathbf{A},\mathbf{R}}  {T^{\mathrm{u}}},  \\
			\rm{s.t.}~	& \sum_{m=0}^{M}  a_{km} = 1, \forall k \in  {\mathcal{P}_j^t}, \label{eqn:constraint:assignment}\\
			&\sum_{k \in K_0} r_{k0}  \leq W^u, \label{eqn:constraint:cloud_bw}\\
			& \sum_{k \in \mathcal{K}_m} r_{km} \leq B^{\mathrm{fr}}_m, \forall m \in \mathcal{M} \setminus \{0\}, \label{eqn:constraint:edge_bw}\\ 
			& a_{km} \in \{0,1\},\label{eqn:constraint:assign_variable}\\
			& r_{k0} \in[0, W^u],~\textrm{and}~ r_{km} \in[0,B^{\mathrm{fr}}_m] , \forall m \in \mathcal{M} \setminus \{0\}. \label{eqn:constraint:bw_variable}
			\end{align}
	\end{subequations}
	The constraints (\ref{eqn:constraint:assignment}) guarantee that a user can associate with only one edge node in one iteration\footnote{A more general network model could be considered. The binary constraints can be replaced by $a_{km} \in [0,1]$. However, in that scenario, the proposed in-network aggregation process at edge nodes could run with incomplete models. Those models are then forwarded to the cloud to be aggregated. Consequently, the network suffers extra latency and traffic as penalties. In practice, as $D$ and $K$ can be significantly large, the penalties make the solution of $a_{km}$ close to binary.}. The constraints (\ref{eqn:constraint:cloud_bw}) and (\ref{eqn:constraint:edge_bw}) ensure that total users' data rates associated with each edge node or the cloud node must
	not exceed its bandwidth capacity.  {Here, the transmission latency between an edge node $m$ and the cloud node is computed as in (\ref{eqn:revised_edge_cloud_latency}).} The mixed integer non-linear programming problem $\mathscr F_1^u$ is actually NP-Hard. We then propose a highly efficient randomized rounding solution for practical implementation in the next section.
%
	\begin{proposition} \label{prop:NP_Hard}
		$\mathscr F_1^u$ is a NP-Hard problem.
	\end{proposition}
	\begin{proof}
		To prove that $\mathscr F_1^u$ is a NP-Hard problem, we first introduce Lemma \ref{lemma:1} and Lemma \ref{lemma:2}. These two lemmas allow us to transform $\mathscr F_1^u$ into an equivalent Integer Linear Programming, which is then proven to be NP-Hard. Hence, $\mathscr F_1^u$ is also NP-Hard.
	\begin{lemma} \label{lemma:1} Given any uplink routing matrix $\mathbf{A}$, with $ |\mathcal{K}_0| = \sum_{k \in \mathcal{K}_0}  a_{k0} >0 $, for problem $\mathscr F_1^u$, at the cloud node, the uplink latency for users associated with the cloud node  satisfies
		\begin{align}
		T^{\mathrm{u}}_0 & = \max_{k \in \mathcal{K}_0} \Bigg\{D\frac{a_{k0} }{r_{k0} } \Bigg \}  \geq \frac{D |\mathcal{K}_0|}{W^u}, \forall k \in \mathcal{K}_0 .
		\end{align}
		Here, the equality happens when $r_{10} = \ldots = r_{|\mathcal{K}_0|0} = \frac{W^u}{|\mathcal{K}_0|}$. 
	\end{lemma} 
	
	\begin{lemma}\label{lemma:2}
		Given any uplink routing matrix $\mathbf{A}$, with $ |\mathcal{K}_m| = \sum_{k \in \mathcal{K}_m}  a_{km} >0$, for problem $\mathscr F_1^u$, at each edge node $m$, the uplink latency for users associated with edge node $m$ satisfies
		\begin{align}
\nonumber		T^{\mathrm{u}}_m  & = \max_{k \in \mathcal{K}_m} \Bigg\{D\frac{a_{km}}{r_{km}} \Bigg \} + \min \Bigg \{\frac{D}{B^{\mathrm{bk}}_m}, \frac{D\sum_{k \in \mathcal{K}_m }  a_{km}}{B^{\mathrm{bk}}_m} \Bigg \}\\
		& \geq \frac{D |\mathcal{K}_m|}{B^{\mathrm{fr}}_m}  + \min \Bigg \{\frac{D}{B^{\mathrm{bk}}_m}, \frac{D|\mathcal{K}_m|}{B^{\mathrm{bk}}_m} \Bigg \}.
		\end{align}
		The equality happens when $r_{1m} = \ldots = r_{|\mathcal{K}_m|m} = \frac{B^{\mathrm{fr}}_m}{|\mathcal{K}_m|} $.
	\end{lemma}		
	Following Lemma \ref{lemma:1} and Lemma \ref{lemma:2}, 	we  {can observe} that the network operator only needs to  optimize the uplink routing matrix while the uplink data rates for users associated with edge nodes or cloud node will be  {equally allocated among their directly associated users.} If there exists $m$ such that $|\mathcal{K}_m| = 0$, we can arbitrarily set the value of $\mathbf{r}_{m}$ and set the value of $T^{\mathrm{u}}_m$ as $0$. As a result, $\mathscr F_1^u$ is reduced to
	\begin{align}
	\nonumber \mathscr F_2^u: & \min_{\mathbf{A}} \max \Bigg\{ \max_m \Bigg\{ D\frac{ \sum_{k \in \mathcal{K}_m}  a_{km}}{B^{\mathrm{fr}}_m} \\
	\nonumber	& \hspace{0.5cm} + \min \Bigg \{\frac{D}{B^{\mathrm{bk}}_m}, \frac{D\sum_{k \in \mathcal{K}_m }  a_{km}}{B^{\mathrm{bk}}_m} \Bigg \}, D \frac{ \sum_{k \in \mathcal{K}_0}  a_{k0} }{W^u}  \Bigg\},\\
	\rm{s.t.}~	& (\ref{eqn:constraint:assignment})~ \textrm{and} ~(\ref{eqn:constraint:assign_variable}),
	\end{align}
	where the optimal solution in $\mathscr F_2^u$ is also the optimal solution in $\mathscr F_1^u$. Consider a special case when the edge-to-cloud latency are negligible, i.e. $B^{\mathrm{bk}}_m \to \infty$. For the convenience of notations, let $B^{\mathrm{fr}}_0 \triangleq W^u$. We have 
	\begin{align}
	\nonumber \mathscr F_2^u: & \min_{\mathbf{A}} \Bigg\{  \max_m \Bigg\{D \frac{ \sum_{k \in \mathcal{K}_m}  a_{km}}{B^{\mathrm{fr}}_m} \Bigg\} \Bigg \} ,\\
	\rm{s.t.}~	& (\ref{eqn:constraint:assignment})~ \textrm{and} ~(\ref{eqn:constraint:assign_variable}),
	\end{align}		
	This formulation is mathematically similar to the makespan minimization problem for parallel machines, which is NP-Hard	\cite{Pinedo2008}, where the makespan is the completion time of the last task. The special case is NP-hard and hence, so is $\mathscr F_2^u$. Since $\mathscr F_2^u$ is an NP-Hard problem, $\mathscr F_1^u$ is also an NP-Hard problem.
	\end{proof}
	
	If the proposed in-network computation protocol is not considered, the aggregation latency-minimized routing framework is formulated as  followed
	\begin{align}
	\nonumber \mathscr Q_1^u: & \min_{\mathbf{A},\mathbf{R}}  {T^{\mathrm{u}}},\\
\rm{s.t.}~	&  (\ref{eqn:constraint:assignment})-(\ref{eqn:constraint:bw_variable}),
	\end{align}
	 {where the transmission latency between an edge node $m$ and the cloud node is computed as in (\ref{eqn:no_mec_edge_cloud_lat})}.
		\begin{proposition}
			$\mathscr Q_1^u$ is a NP-Hard problem
		\end{proposition} 
		\begin{proof}
		Similar to the proof of Proposition \ref{prop:NP_Hard}, we first introduce Lemma \ref{lemma:3} and Lemma \ref{lemma:4}. The proofs of these lemmas are similar to those of Lemma \ref{lemma:1} and Lemma \ref{lemma:2}, hence omitted.
	\begin{lemma} \label{lemma:3} Given any uplink routing matrix $\mathbf{A}$, with $ |\mathcal{K}_0| = \sum_{k \in \mathcal{K}_0}  a_{k0} >0 $, for problem $\mathscr Q_1^u$, at the cloud node, we observe that $T^{\mathrm{u}}_0$ satisfies
		\begin{align}
		T^{\mathrm{u}}_0 & = \max_{k \in \mathcal{K}_0} \Bigg\{D\frac{a_{k0} }{r_{k0} } \Bigg \}  \geq \frac{D |\mathcal{K}_0|}{W^u}, \forall k \in \mathcal{K}_0 .
		\end{align}
		Here, the equality happens when $r_{10} = \ldots = r_{|\mathcal{K}_0|0} = \frac{W^u}{|\mathcal{K}_0|}$. 
	\end{lemma}	
	\begin{lemma} \label{lemma:4} Given any uplink routing matrix $\mathbf{A}$, with $ |\mathcal{K}_m| = \sum_{k \in \mathcal{K}_m}  a_{km} >0$, for problem $\mathscr Q_1^u$, at each edge node $m$, we observe that $T^{\mathrm{u}}_m$ satisfies
		\begin{align}
		T^{\mathrm{u}}_m  & = \max_{k \in \mathcal{K}_m} \Bigg\{D\frac{a_{km}}{r_{km}} \Bigg \} + \frac{D|\mathcal{K}_m|}{B^{\mathrm{bk}}_m} \geq \frac{D |\mathcal{K}_m|}{B^{\mathrm{fr}}_m}  + \frac{D|\mathcal{K}_m|}{B^{\mathrm{bk}}_m},
		\end{align}
		where $ |\mathcal{K}_m| = \sum_{k \in \mathcal{K}_m}  a_{km}$.	The equality happens when $r_{1m} = \ldots = r_{|\mathcal{K}_m|m} = \frac{B^{\mathrm{fr}}_m}{|\mathcal{K}_m|} $.
	\end{lemma}	
	Following Lemma \ref{lemma:3} and Lemma \ref{lemma:4}, we see that the network operator  {needs} to only optimize the uplink routing matrix, while the uplink data rates for users associated with edge nodes or cloud node will be  {equally allocated}.  As a result, $\mathscr Q_1^u$ is reduced to
	\begin{align}
	\nonumber \mathscr Q_2^u: & \min_{\mathbf{A}} \max \Bigg\{ \max_m \Bigg\{ D\frac{ \sum_{k \in \mathcal{K}_m}  a_{km}}{B^{\mathrm{fr}}_m} + D\frac{\sum_{k \in \mathcal{K}_m}  a_{km}}{B^{\mathrm{bk}}_m} \Bigg\}, \\
\nonumber & \hspace{1.7cm}	D\frac{ \sum_{k \in \mathcal{K}_0}  a_{k0} }{W^u}  \Bigg\},\\
	\rm{s.t.}~	& (\ref{eqn:constraint:assignment})~ \textrm{and} ~(\ref{eqn:constraint:assign_variable}) ,
	\end{align}
	where the optimal solution in $\mathscr Q_2^u$ is also the optimal solution in $\mathscr Q_1^u$. Similar to Proposition \ref{prop:NP_Hard},  we consider a special case when the edge-to-cloud latency are negligible, i.e. $B^{\mathrm{bk}}_m \to \infty$. Here, we see that $\mathscr Q_2^u$ is a NP-Hard problem, so is $\mathscr Q_1^u$.
		\end{proof}

	\section{Randomized Rounding Based Solution} \label{section:solution}
	\subsection{With In-Network Computation Protocol}
	In this section, we present an approximation algorithm for the main problem that leverages a randomized rounding technique. The proposed algorithm is summarized in Algorithm \ref{algorithm:heu_per_plan}. First, we introduce auxiliary variables $y$ and $\gamma_m$ into $\mathscr F_2^u$ such that
	\begin{align}
	y & \geq \max \Bigg \{ \max_m \Bigg\{D \frac{ \sum_{k \in \mathcal{K}_m}  a_{km}}{B^{\mathrm{fr}}_m} + \gamma_m \Bigg\}, D\frac{ \sum_{k \in \mathcal{K}_0}  a_{k0} }{W^u}  \Bigg \},\\
	\gamma_m & \leq \min \Bigg \{\frac{D}{B^{\mathrm{bk}}_m}, \frac{D\sum_{k \in \mathcal{K}_m }  a_{km}}{B^{\mathrm{bk}}_m} \Bigg \}, \forall m
	\end{align}
	Problem $\mathscr F_2^u$ is then equivalently rewritten as
	\begin{subequations}
	\begin{align}
	\nonumber \mathscr F_3^u: & \min_{y,\{\gamma_m\}, \{a_{km}\}} y, \\
	\mathrm{s.t.}~& y \geq D\frac{ \sum_{k \in \mathcal{K}_m}  a_{km}}{B^{\mathrm{fr}}_m} + \gamma_m, \forall m, \label{eqn:aux_cons:inc:edge_lat}\\
	& y \geq D\frac{ \sum_{k \in \mathcal{K}_0}  a_{k0} }{W^u}, \label{eqn:aux_cons:inc:cloud_lat}\\
	& \gamma_m \leq \frac{D}{B^{\mathrm{bk}}_m}, \forall m \in \mathcal{M} \setminus \{0\}, \label{eqn:aux_cons:inc:ina} \\
	& \gamma_m \leq D\frac{ \sum_{k \in \mathcal{K}_0}  a_{k0} }{W^u} \label{eqn:aux_cons:inc:non_ina}, \forall m \in \mathcal{M} \setminus \{0\}, \\
	& (\ref{eqn:constraint:assignment})~ \textrm{and} ~(\ref{eqn:constraint:assign_variable}). \nonumber
	\end{align}
	\end{subequations}
	The Algorithm \ref{algorithm:heu_per_plan} starts by solving the Linear Relaxation (LR) of $\mathscr F_3^u$. Specifically, it relaxes the variables
	$a_{km}$ to be fractional, rather than integer. The
	Linear Relaxation of $\mathscr F_3^u$ can be expressed as follows:
			\begin{align}
			\nonumber \mathscr F_4^u: & \min_{y,\{\gamma_m\}, \{a_{km}\}} y,  \\
			\mathrm{s.t.}~& (\ref{eqn:aux_cons:inc:edge_lat})-(\ref{eqn:aux_cons:inc:non_ina}), (\ref{eqn:constraint:assignment}),~\textrm{and}~ a_{km} \in [0,1].
			\end{align}
	
	Let $z^{\dagger}= [{\tilde{\mathbf{a}}}^{\dagger},y^{\dagger}, \gamma_1^{\dagger}, \ldots,\gamma_M^{\dagger}]$ denote the optimal solution of $\mathscr F_4^u$.
	We first transform ${\tilde{\mathbf{a}}}^{\dagger}$ to the equivalent fractional matrix ${\mathbf{A}}^{\dagger}$,
	whose elements are in $[0, 1]$. If all components
	of ${\mathbf{A}}^{\dagger}$ are binary, it is also the optimal solution to $\mathscr F_2^u$. Otherwise,
	to recover binary characteristic of ${\mathbf{A}}$, for each row of ${\mathbf{A}}^{\dagger}$,
	we randomly round the element $a_{km}$ to 1 with probability ${a_{km}}^{\dagger}$. The decision is done in an exclusive manner for satisfying constraint (\ref{eqn:constraint:assignment}). It means that for each row $k$, only one element of the row is one, the rest are zeros.  For example, let us consider $2$ edge nodes and at row $k$,  assume $a_{k0} = 0.7$, $a_{k1} = 0.1$ and $a_{k2} = 0.2$. We construct $3$ intervals, $0: [0, 0.7], 1: (0.7, 0.8]$ and $2: (0.8,1]$. We then randomly pick a number uniformly distributed in $[0,1]$. If the value is in interval $1$, we set $a_{k1} =1$, and the rest are zeros. The random decision is made independently for each $k$. Following
	this procedure, we obtain the solution ${\mathbf{A}}^{(\rm{Alg})}$. Then, ${r_{1,0} }^{(\rm{Alg})} = \ldots = {r_{|\mathcal{K}_0|,0} }^{(\rm{Alg})} = \frac{W^u}{|\mathcal{K}_0|}$ and ${{r_{1m}} }^{(\rm{Alg})} = \ldots = {r_{|\mathcal{K}_m|m} }^{(\rm{Alg})} = \frac{B^{\mathrm{fr}}_m}{|\mathcal{K}_m|}, \forall m$. The complexity of this algorithm is $O(\nu^{3.5}\Omega^2)$, where $\nu = K(M+1)+1$,  {$\Omega$ is the number of the bits in the input \cite{nesterov1994interior}.}  {Here, we provide the guarantee on the quality of the aggregation latency returned by Algorithm \ref{algorithm:heu_per_plan}.}
	
	\begin{algorithm}    [t]                
		\caption{Randomized Routing Algorithm for Low Latency Federated Learning}          
		\label{algorithm:heu_per_plan}                           
		{\footnotesize		
			\begin{algorithmic}[1]                    
				\Require $D$, $B^{\mathrm{fr}}_m$, $B^{\mathrm{bk}}_m$, $W^u$ and $\mathcal{K}^t$.
				\Ensure ${\mathbf{A} }^{(\rm{Alg})}$, ${\mathbf{R} }^{(\rm{Alg})}$
				\State Solve $\mathscr F_4^u$  to achieve ${\mathbf{A}}^{\dagger}$.
				\If  {${\mathbf{A}}^{\dagger}$ is binary}
				\State ${\mathbf{A} }^{(\rm{Alg})} = {\mathbf{A}}^{\dagger}$
				\Else
				\For {$k = 1$ to $k= K$}
				\State {${a_{km}}^{(\rm{Alg})} = 1$ with probability ${a_{km}}^{\dagger}$} with exclusive manner based on constraint (\ref{eqn:constraint:assignment})
				\EndFor
				\EndIf						
				
				\State {Then,
					\begin{align}
					\nonumber	& {r_{1,0} }^{(\rm{Alg})} = \ldots = {r_{|\mathcal{K}_0|,0} }^{(\rm{Alg})} = \frac{W^u}{|\mathcal{K}_0|}\\
					\nonumber 	&	{{r_{1m}} }^{(\rm{Alg})} = \ldots = {r_{|\mathcal{K}_m|m} }^{(\rm{Alg})} = \frac{B^{\mathrm{fr}}_m}{|\mathcal{K}_m|}, \forall m ~\textrm{such that}~|\mathcal{K}_m| \neq 0.
					\end{align}
				}
			\end{algorithmic}}
		\end{algorithm}
	
		\begin{theorem}
			The aggregation latency returned by Algorithm \ref{algorithm:heu_per_plan} is at most $\frac{2 \ln K}{y^{\dagger}} + 3$ times higher than the optimal with
			high probability  {$1 - 1/K$}, under the assumption $y^{\dagger} > \ln (K)$,  {where $y^{\dagger}$ is the optimal value of the $\mathscr F_4^u$ and also the lower bound of the optimal, and $K$ is the number of users.}
		\end{theorem}
		\begin{proof}
			Let $y^{\dagger}$ denote the optimal value of the $\mathscr F_4^u$ and $U = 6 \ln(K)y^{\dagger}$. For $\delta >0$, applying the Chernoff bound, we have
			\begin{align}
			\mathrm{Pr} \Bigg [T^{\mathrm{u}}_m   > (1+ \delta) y^{\dagger} \Bigg] & \leq e^{-\frac{\delta^2y^{ \dagger}}{2+ \delta}} .
			\end{align}
			Next, we upper bound the right hand side of the above
			inequality by $1/K^2$. In order  {to achieve this condition,} the $\delta$ value
			must satisfy the following condition:
			\begin{align}
			\delta \geq \frac{\ln (K)}{y^{ \dagger}} + \sqrt{\frac{\ln^2(K)}{y^{ \dagger 2}}+\frac{4\ln(K)}{y^{ \dagger}}},
			\end{align}
			with the assumption $y^{\dagger} > \ln (K)$. The above condition holds if we pick $\delta = \frac{2\ln (K)}{y^{ \dagger}} +2$. Then, by applying the union bound, we get
			\begin{align}
\nonumber			\mathrm{Pr}[ {\exists m |} T^{\mathrm{u}}_m >(1+ \delta) y^{\dagger}]	& \leq \sum \mathrm{Pr}[T^{\mathrm{u}}_m >(1+ \delta) y^{\dagger} ]   \\
			& \leq \frac{M+1}{K^2} \leq \frac{1}{K}.
			\end{align}
							\begin{figure}[t]
					\centering
					{\includegraphics[width=\linewidth]{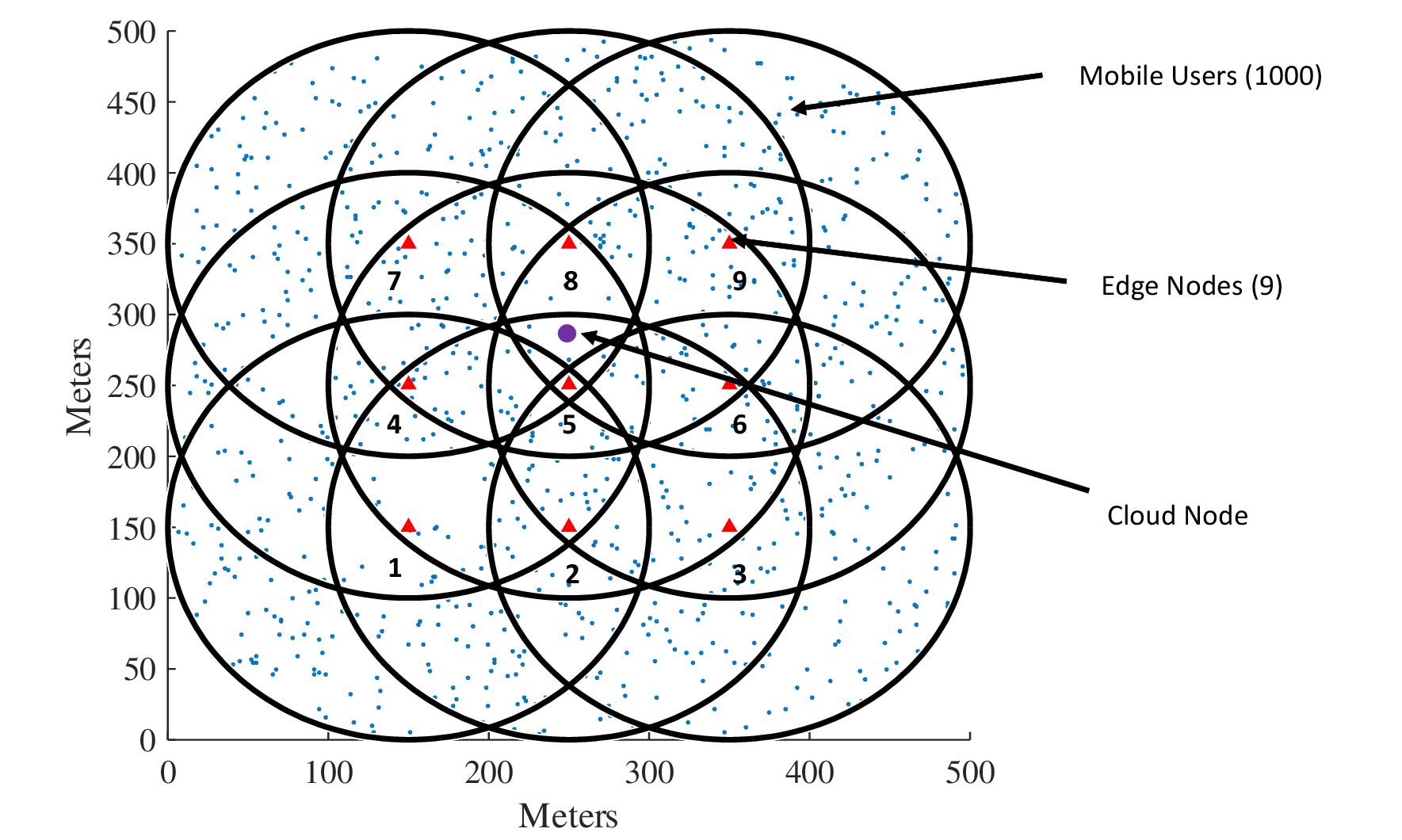}}
					\caption{The network setup. \label{fig::network_setup} }
				\end{figure}
			Consequently, with high probability  {$1 - \frac{1}{K}$}, the resulting aggregation latency is at most $1 + \delta = \frac{2 \ln K}{y^{\dagger}} + 3$ times worse
			than that of the optimal solution.
		\end{proof}
		 {In practice, the number of users $K$ is large, for example $5000$. It is reasonable to assume that the lower bound of latency is higher than $\ln K = 8.5$ s. In the simulations below, we observe that the assumption holds for practical settings.}
		\subsection{Without In-Network Computation Protocol}
		In this section, we  {aim} to find the lower bound of network latency when the proposed in-network computation protocol is not considered. First, we introduce an auxiliary variable $y$ such that
		\begin{align}
\nonumber		y \geq \max \Bigg \{ & \max_m \Bigg\{D \frac{ \sum_{k \in \mathcal{K}_m}  a_{km}}{B^{\mathrm{fr}}_m} + D\frac{\sum_{k \in \mathcal{K}_m}  a_{km}}{B^{\mathrm{bk}}_m} \Bigg\}, \\
		& D\frac{ \sum_{k \in \mathcal{K}_0}  a_{k0} }{W^u}  \Bigg \}.
		\end{align}
		Hence, problem $\mathscr Q_2^u$ is equivalently transformed to
		\begin{subequations}
					\begin{align}
					\nonumber \mathscr Q_3^u: & \min_{y, \{a_{km}\}} y,  \\
					\mathrm{s.t.}~& y \geq D\frac{ \sum_{k \in \mathcal{K}_m}  a_{km}}{B^{\mathrm{fr}}_m} + D\frac{\sum_{k \in \mathcal{K}_m}  a_{km}}{B^{\mathrm{bk}}_m}, \forall m, \label{eqn:aux_cons:non_inc:edge_lat}\\
					& y \geq D\frac{ \sum_{k \in \mathcal{K}_0}  a_{k0} }{W^u},\label{eqn:aux_cons:non_inc:cloud_lat}\\
					& (\ref{eqn:constraint:assignment})~ \textrm{and} ~(\ref{eqn:constraint:assign_variable}). \nonumber
					\end{align}
		\end{subequations}
	Like Algorithm \ref{algorithm:heu_per_plan}, we relax the variables
		$a_{km}$ to be fractional, rather than integer. The
		Linear Relaxation formulation of $\nonumber \mathscr Q_3^u$ can be expressed as follows:
			\begin{align}
			\nonumber \mathscr Q_4^u: 	& \min_{y,\{a_{km}\}} y , \\
			\mathrm{s.t.}~& (\ref{eqn:aux_cons:non_inc:edge_lat}), (\ref{eqn:aux_cons:non_inc:cloud_lat}),(\ref{eqn:constraint:assignment}),
			~\text{and}~ a_{km} \in [0,1].
			\end{align}
		Let $y^{\dagger \dagger}$ denote the optimal solution of $\mathscr Q_4^u$. $y^{\dagger \dagger}$ is the lower bound of the uplink latency without the proposed in-network computation protocol.  {We use $y^{\dagger \dagger}$ to proxy the aggregation latency when INA is not considered.}

		\section{Numerical Results} \label{section:numerical}

		In this section, we carry out  {simulations} to evaluate the
		performance of the proposed frameworks. We consider a
		similar setup as in \cite{PoularakisTON2020}, depicted in  {Fig.} \ref{fig::network_setup}.
		Here, $M = 9$ edge nodes are regularly deployed on a
		grid network inside a $500 \times 500 ~\rm{m}^2$ area. $K = 1000$ mobile
		users are uniformly distributed over the edge nodes' coverage regions (each of $150$m radius). All mobile users are within the coverage of the cloud node. For each edge node $m$, we set the uplink
		fronthaul capacity to $B^{\mathrm{fr}}_m = 1$Gbps, the backhaul capacity to $B^{\mathrm{bk}}_m = 1$Gbps. These settings are inspired by WiFi IEEE 802.11ac standards \cite{Wifiax}, and data centers interconnection using optical fibers \cite{Chengthesis2019}. We also set the cloud uplink and downlink capacities $W^u = W^d = 2$Gbps.  {As in \cite{TranINFOCOM2019}, we set $t_{\min}^{\rm{cp}}$ and $t_{\max}^{\rm{cp}}$ at $0.2$s and $80$s, respectively.} For model aggregation, by default, we investigate our system using ResNet152's model size, i.e., $D = 232$ MB \cite{KerasPretrain}. In later simulations, we also investigate our system with different model sizes.
		
					\begin{table}[t] \footnotesize			
						\caption{Default Parameter Setup.} \label{tab:para_setup_2}
						\begin{center}
							\begin{tabular}{|c|c| }
								\hline
								{\bf Parameter} & {}{\bf Value} \\  \hline
								$M$ & $9$  \\  
								$K$ & $1000$   \\  
								EN's coverage & $150$ m \\  
								EN to EN distance & $100$ m \\  
								$B^{\mathrm{fr}}_m$ & $1$Gbps \\  
								$B^{\mathrm{bk}}_m$ & $1$Gbps\\  
								$W^u$ & $2$ Gbps \\
									$ W^d$ & $2$ Gbps \\
								$t_{\min}^{\rm{cp}}, t_{\max}^{\rm{cp}} $& $0.2$ s, $8$s \\  
									$ D$ & $232$ MB \\
								\hline	
							\end{tabular}
						\end{center}
					\end{table}%
		\subsection{Algorithm Comparison - Latency Reduction}
				\begin{figure}[t]
					\centering
					{\includegraphics[width=0.8\linewidth]{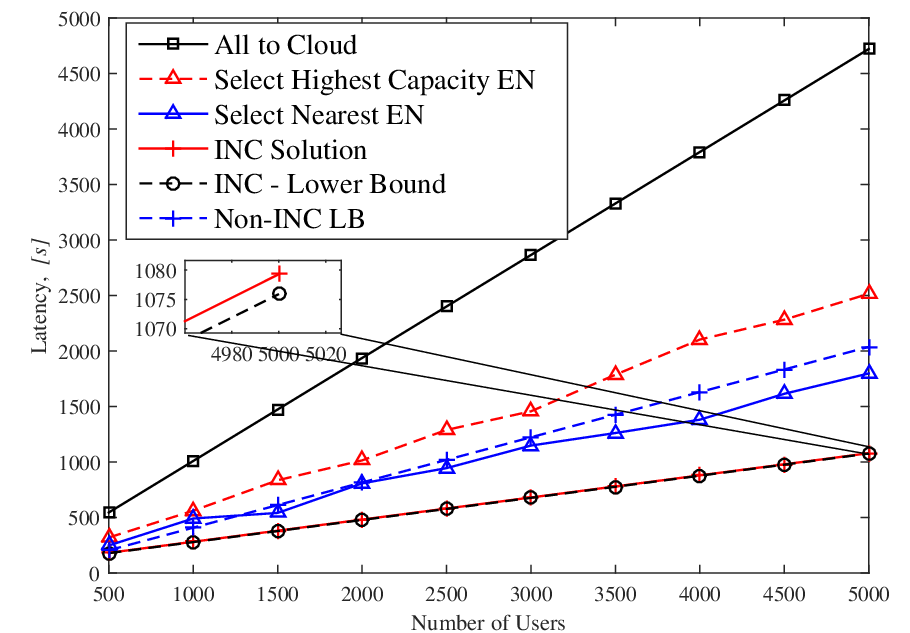}}
					\caption{Algorithm comparison with respect to different number of users. \label{fig::alg_com_change_K} }
				\end{figure}
		Fig. \ref{fig::alg_com_change_K} compares the one-iteration  {training time} of different algorithms
		versus the number of users $K$. The proposed INC protocol is
		compared with three other baseline methods, namely:
		\begin{itemize}
			\item[1.] ``Only Cloud'': All $K$ users' models are directly aggregated via their direct link to the cloud. Here, models are not forwarded to any edge nodes (ENs).
			\item[2.] ``Select Highest Capacity EN'': In this greedy-manner framework, each user selects the neighbor EN with the highest uplink data rate capacity. 
			\item[3.] ``Select Nearest EN'': In this  greedy-manner  framework, each user selects the nearest neighbor EN. 
			\item[4.] ``INC Solution'': $K$ users can associate with the cloud node and edge nodes with INC protocol.  The network routing problem $\mathscr F_3^u$ is solved by using Algorithm \ref{algorithm:heu_per_plan}.
			\item[5.] ``INC-Lower Bound'': In this scenario, we use Linear
			Relaxation to solve $\mathscr F_3^u$. This scenario will provide the lower bound of network latency if the proposed INC protocol is considered.
			\item[6.] ``Non-INC LB'': $K$ users can associate with the cloud node and edge nodes without the INC protocol. By using Linear
			Relaxation to solve $\mathscr Q_3^u$, this method will provide the lower bound of network latency if INA process is not implemented at edge nodes and the cloud node.
		\end{itemize}
		
		Here, users are scheduled by following Algorithm \ref{algorithm:user_scheduling}. As can be  {observed in} Fig. \ref{fig::alg_com_change_K}, our proposed algorithm can achieve near optimal performance. When $K = 5000$, the
		latency obtained by the proposed solution is approximately $0.7\%$ higher than that of the \textit{INC LB}. It implies that our
		proposed solution can achieve the performance almost the
		same as that of the lower bound solution even with a high number of users. \textit{Only Cloud} has the worst performance. For example, when $K = 5000$, the network latency of \textit{Only Cloud} is $4721$s which is  {is $1.8$ times higher} than that of the second worst one, \textit{Select Highest Capacity EN}, $2516$s. The second framework, ``Select Highest Capacity EN'', provides the second lowest performance due to resource contention. Since users select the highest EN with the highest capacity EN of their neighbor EN set, user locating in the overlapping region of the same set of ENs will select the same EN. As a result, there is a resource contention in the chosen EN. The other greedy-manner, ``Select Nearest EN'', can achieve better performance. Because users' positions are randomly distributed, the number of users in each EN's coverage is approximately the same. Thus, this framework can achieve an approximated load-balancing solution which is significantly better than the second framework. Last but not least, the  {gaps} between our proposed algorithm with other baselines increase as the number of users $K$ increases.  {This clearly shows} that our  {proposed solution} is  {significantly beneficial} for very large scale federated learning networks.
		
		Since our algorithm requires to compute the Linear Relaxation results before conducting randomization, its time complexity is higher than those of the baseline methods. Table \ref{tab:algo_complex} summarizes the time complexity of the four frameworks.
		\begin{table}[t] \footnotesize			
			\caption{Time Complexity.} \label{tab:algo_complex}
			\begin{center}
				\begin{tabular}{|c|c| }
					\hline
					{\bf Frameworks} & {}{\bf Complexity} \\  \hline
					``All to Cloud'' & $O(1)$  \\  
					``Select Highest Capacity EN'' & $O(KM)$   \\  
					``Select Nearest EN'' & $O(KM)$ \\  
					``INC Solution'' &$O(\nu^{3.5}\Omega^2)$\\  
					\hline	
				\end{tabular}
			\end{center}
		\end{table}%

		\begin{figure}[t]
			\centering
			{\includegraphics[width=0.9\linewidth]{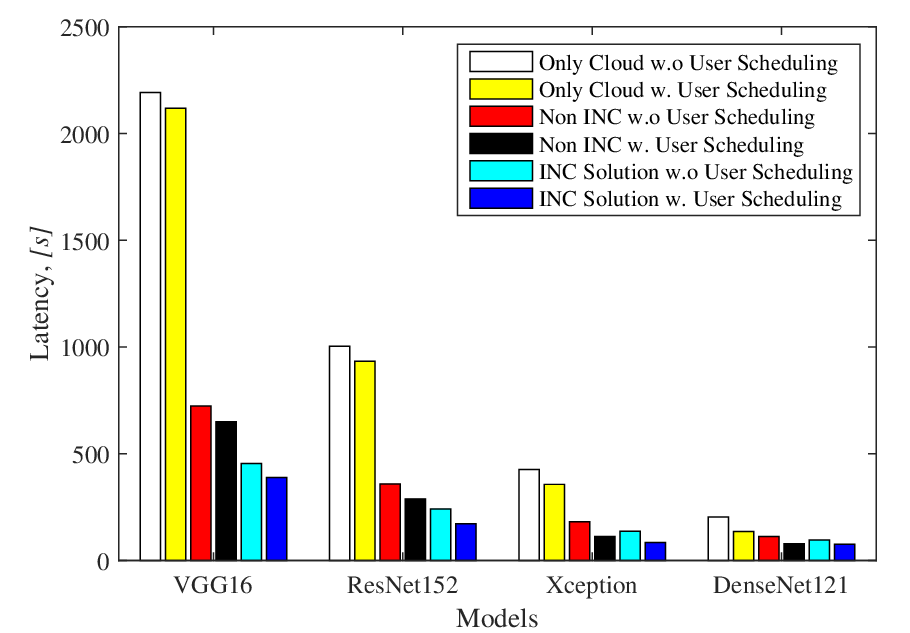}}
			\caption{Algorithm comparison with respect to different models. \label{fig::alg_com_change_D} }
		\end{figure}

		 In Fig. \ref{fig::alg_com_change_D}, consider \textit{Only Cloud}, \textit{Non-INC LB} and  our \textit{INC solution}, we evaluate the network latency with or without user scheduling mechanism in different models. In the case when the user scheduling mechanism is not considered, all users need to wait until the slowest user finishing it local processing step. They are VGG16, ResNet152, Xception and DenseNet121 whose model sizes are $528$MB, $232$MB, $88$MB and $33$MB, respectively. The distribution of $p(t^{\rm{cp}})$ is remained unchanged since we want to focus only on the variation of aggregation latency as $D$ changes. Here, we choose the default setting  {with} $K=1000$. We observe that our user scheduling mechanism can mitigate straggler effects due to slow workers. For example, when we consider ResNet152 model, \textit{INC solution} saves $28.49\%$ in comparison with the case where user scheduling mechanism is not considered. Moreover, since the aggregation latency decreases as the model's size $D$ decreases, we observe that the latency saving of \textit{INC solution w.o User Scheduling} to \textit{INC solution w. User Scheduling} increases then decreases. For example, the saving is $14.63\%$, $28.49\%$, $38.18$ and $20.43\%$ for VGG16, ResNet152, Xception and DenseNet121, respectively. It implies that as the communication latency decreases, the contribution of the slowest node's computing delay increases. Until a certain value $D$,  the slowest node's computing delay, which is fixed in this simulation, becomes the major part of the whole network delays. Hence, the saving decreases again.
		\subsection{Impacts of $\Delta t$}
		In this subsection, we investigate the impacts of $\Delta t$ on the latency of one learning round. Here, we used the default setting with number of users $K = 1000$. As we can see in the above figure, as $\Delta t$ increases, the number of users in $P_1$ increases because all users with $t_{k}^{\rm{cp}} \leq t_{\min}^{\rm{cp}}+\Delta t$ are assigned to $P_1$. It is also the reason why the number  of users in $P_2$ decreases. As can be seen in Fig. \ref{fig::change_delta_t}, we observe that with small $\Delta t$, most of users wait  until $ t_{\max}^{\rm{cp}}$ to be collected in partition $P_2$. Thus, we can see high latency with small $\Delta t$. As $\Delta t$ increases, more users go to $P_1$, thus the latency decreases. When $\Delta t$ is big enough, the latency rises again because most of users are in $P_1$ and their local models have to wait after $t_{\min}^{\rm{cp}}+\Delta t$ to be collected. When $\Delta t =t_{\max}^{\rm{cp}} = 80$s, the latency is highest since all users' models start to be collected after $t_{\max}^{\rm{cp}}$.
			\begin{figure}[t]
				\centering
				{\includegraphics[width=\linewidth]{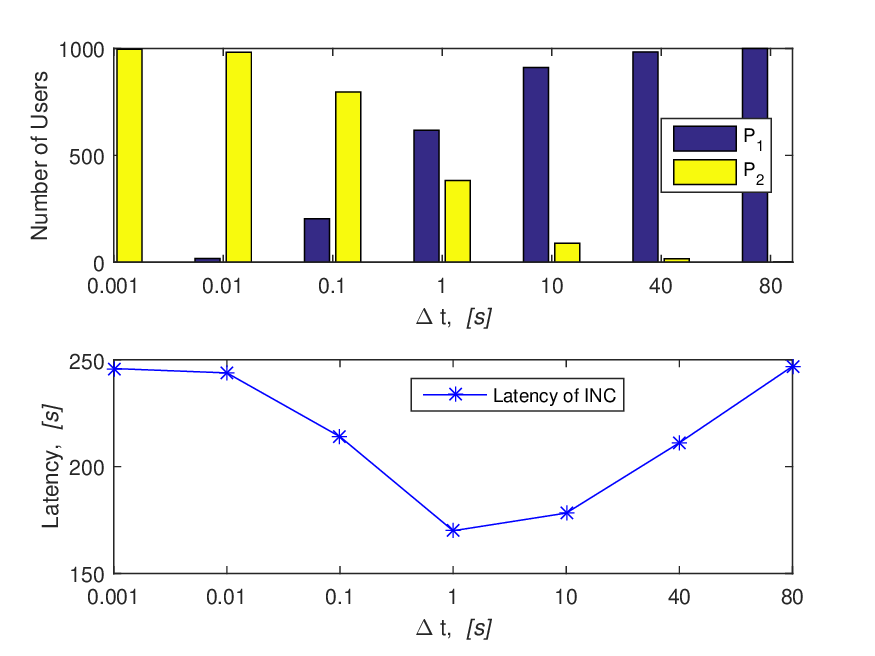}}
				\caption{Latency of one learning round and the number of users in each partition change w.r.t $\Delta t$. \label{fig::change_delta_t} }
			\end{figure} 
		\subsection{Traffic and Computation Reduction at the Cloud Node}
		
		\begin{figure}[h]
			\centering
			\subfigure[]{\includegraphics[width=0.8\linewidth]{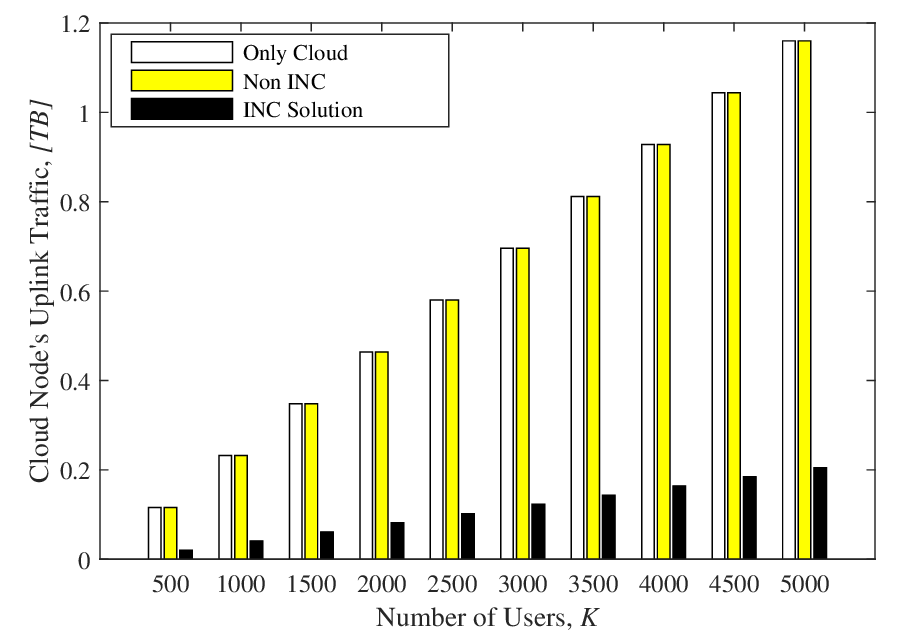}}
			\subfigure[]{\includegraphics[width=0.8\linewidth]{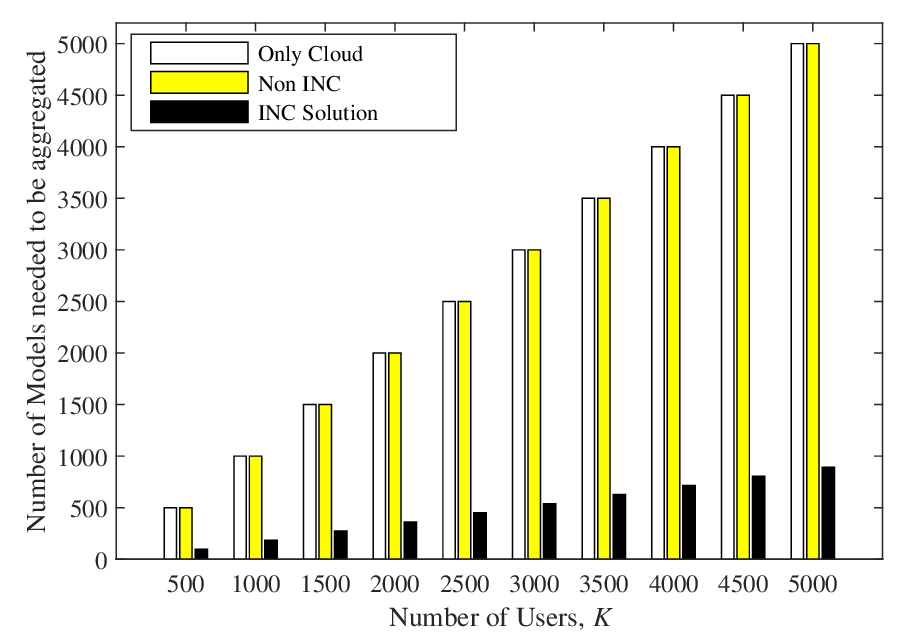}}
			\caption{Cloud node's uplink traffic and computing load. \label{fig::network_load} }
		\end{figure}
		In this subsection, using ResNet152's model setting, we investigate the uplink traffic and the number of models needed to be aggregated at the cloud node in one learning iteration. We compare three schemes: \textit{Only Cloud}, \textit{Non-INC LB} and  our \textit{INC solution}. The number of models needed to be aggregated at the cloud node is proportional to the number of computations here. In Figs. \ref{fig::network_load}, the uplink traffic and the number of computations of \textit{Non-INC LB} at the cloud node is equal to those of \textit{Only Cloud}. It is because all models need to be sent to the cloud before being aggregated and edge nodes only  {forward} the models from users to the cloud without INA process. Meanwhile, with \textit{INC solution}, the two metrics are significantly reduced. For example, when $K = 5000$, the traffic is $0.2$TB for our scheme and $1.16$TB for the two others. Our scheme achieves more than $5$ times lower traffic than the others.  In this simulation, with ResNet152, the number of parameters of each local models is $60,419,944$. As a result, without the proposed INA process, the cloud node needs to aggregate $K$ models whose sizes are more $60$ millions elements. This will  {consume a huge amount of} processing and memory resources. As can be seen, the \textit{INC solution} can reduce the number of models needed to be aggregated at the cloud by more than $5$ times.

		\subsection{Impacts of the number of additional edge nodes' connectivities}

		In this subsection, we investigate the impacts of the number extra edge nodes users on the straggler effects due to bad communication links. 
In this simulation, our system suffers straggler effects when a user's  all wireless connections are bad. As a result, its model is not able to be aggregated at the cloud node. Without loss of generality, we assume that the networking components of servers, edge nodes or cloud node, have probabilities of being faulted. We define $p_{\rm{cloud}}$ as the outage probability of the cloud. Similarly, we define $p_{\rm{edge}}$  as the outage probability of a given edge node (here we assume all edge nodes have the same outage probability). Thus, the probability that the system suffers the straggler effect ${P_s}$ is computed as
		\begin{align}
		\mathrm{P_s} =  p_{\rm{cloud}}p_{\rm{edge}}^{v},
		\end{align}
		where $v$ is the number of extra edge nodes users can connect.
		As can be  {observed in} Fig. \ref{fig::straggle}, increasing the number of additional edge connectivities  {can significantly mitigate} straggler effects. With $p_{\rm{cloud}} = 0.3$, we can decrease the straggle effects $4$ times and $12$ times by providing two additional edge connections for users, with the well-being probabilities of edge node's networking component are $0.5$ and $0.7$, respectively.
		\begin{figure}[t]
			\centering
			{\includegraphics[width=0.8\linewidth]{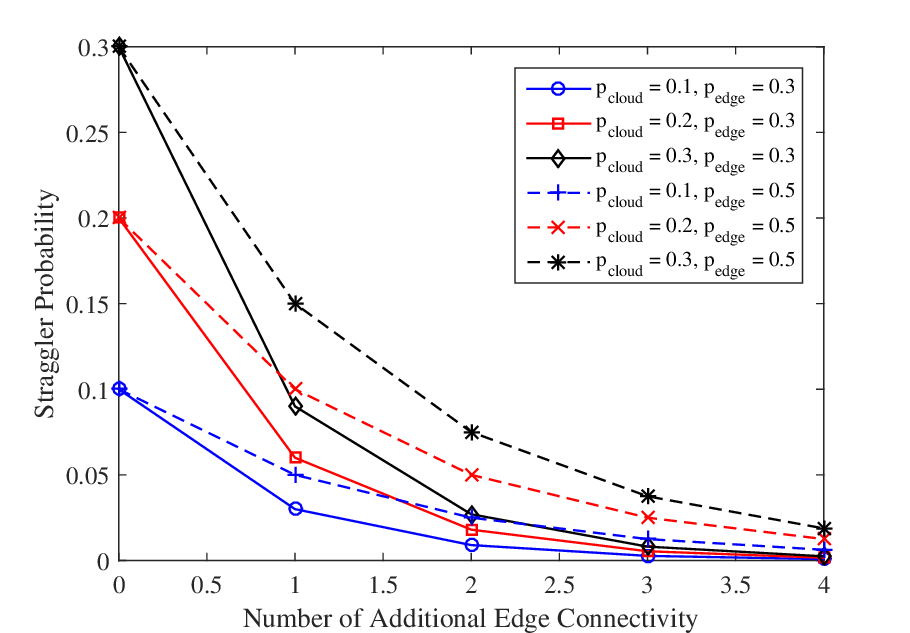}}
			\caption{Straggler effect's probability. \label{fig::straggle} }
		\end{figure}		
		

		\section{Conclusion and Future Directions} \label{section:conclusion}
		In this article, we proposed a  {novel} edge network architecture to decentralize the communications and computing burden of cloud node in Federated Learning. To that end, we designed an in-network computation protocol (INC)  consisting of a user scheduling mechanism, an in-network aggregation process (INA), and an routing algorithm. The in-network aggregation process, which is implemented at edge nodes and cloud node, can adapt two typical methods to solve  the distributed machine learning problems. Under the proposed in-network aggregation (INA) framework, we then formulated a joint routing and resource optimization problem, aiming to minimize the aggregation latency. The problem is proved to be NP-Hard. We then derived its near-optimal solution using random rounding with proven performance guarantee. Simulation results {showed} that the proposed algorithm  {can achieve} more than 99 $\%$ of the optimal solution and significantly outperforms all other baseline schemes without INA. The proposed scheme becomes even more effective (in reducing the latency and straggler effects) when more edge nodes are available. Moreover, we also showed that  {the INA} framework not only help reduce training latency in FL but also reduce significantly reduce the traffic load to the cloud node. By embedding the computing/aggregation tasks at edge nodes and leveraging the multi-layer edge-network architecture, the INA framework can enable large-scale FL.


		\ifCLASSOPTIONcaptionsoff
		\newpage
		\fi

		\bibliographystyle{IEEEtran}
		\bibliography{IEEEabrv,references}

	\end{document}